\documentclass[conference]{IEEEtran}
\usepackage{cite}
\usepackage{amsthm,amsmath,amssymb}
\usepackage{mathrsfs}
\usepackage{amsmath,amssymb,amsfonts}
\usepackage{algorithmic}
\usepackage{graphicx}
\usepackage{xcolor}
\usepackage{array}
\usepackage{multirow}
\usepackage{amsfonts}
\usepackage{booktabs}
\usepackage{float}
\usepackage{subfigure}
\usepackage{tikz}
\usepackage{filecontents}
\usepackage{subfigure}
\usepackage{pifont}
\usepackage{caption, subcaption}
\usepackage{xspace}
\usepackage{setspace}

\usepackage [
    n, 
    advantage,
    operators,
    sets,
    adversary,
    landau,
    probability,
    notions,
    logic,
    ff, 
    mm,
    primitives,
    events,
    complexity,
    oracles,
    asymptotics,
    keys
]{cryptocode}

\usepackage{xurl}

\newcommand{\defref}[1]{Definition \ref{#1}}
\newcommand{\figref}[1]{Figure \ref{#1}}
\newcommand{\tabref}[1]{Table \ref{#1}}

\newcommand{\secref}[1]{Section \ref{#1}}

\newtheorem{theorem}{\bf Theorem}
\newtheorem{definition}{\bf Definition}

\createpseudocodeblock{pcb}{center, boxed }{}{}{}

\begin{document}

\title{\huge Data Availability and Decentralization: New Techniques for zk-Rollups in Layer 2 Blockchain Networks}


\author{
	\IEEEauthorblockN{
		Chengpeng Huang\IEEEauthorrefmark{1}, 
		Rui Song\IEEEauthorrefmark{1}, 
		Shang Gao\IEEEauthorrefmark{1}, 
		Yu Guo\IEEEauthorrefmark{2}, 
		Bin Xiao\IEEEauthorrefmark{1}} 
	\IEEEauthorblockA{\IEEEauthorrefmark{1}Hong Kong Polytechnic University\\ Email: chehuang@polyu.edu.hk, csrsong@comp.polyu.edu.hk,  shanggao@comp.polyu.edu.hk, csbxiao@comp.polyu.edu.hk} 
        \IEEEauthorblockA{\IEEEauthorrefmark{2}SECBIT Labs\\ Email: yu.guo@secbit.io}   	 	
}

\maketitle

\begin{abstract}
The scalability limitations of public blockchains have hindered their widespread adoption in real-world applications. While the Ethereum community is pushing forward in zk-rollup (zero-knowledge rollup) solutions, such as introducing the ``blob transaction'' in EIP-4844, Layer 2 networks encounter a data availability problem: storing transactions completely off-chain poses a risk of data loss, particularly when Layer 2 nodes are untrusted. Additionally, building Layer 2 blocks requires significant computational power, compromising the decentralization aspect of Layer 2 networks.

This paper introduces new techniques to address the data availability and decentralization challenges in Layer 2 networks. To ensure data availability, we introduce the concept of ``proof of download'', which ensures that Layer 2 nodes cannot aggregate transactions without downloading historical data. Additionally, we design a ``proof of storage'' scheme that punishes nodes who maliciously delete historical data. For decentralization, we introduce a new role separation for Layer 2, allowing nodes with limited hardware to participate. To further avoid collusion among Layer 2 nodes, we design a ``proof of luck'' scheme, which also provides robust protection against maximal extractable value (MEV) attacks. Experimental results show our techniques not only ensure data availability but also improve overall network efficiency, which implies the practicality and potential of our techniques for real-world implementation.

\end{abstract}

\section{Introduction}
The popularity of public blockchains has surged in recent decades. However, with the increasing of transaction volumes, scalability has emerged as a critical challenge for public blockchains, limiting the potential blockchain applications \cite{karame2016security}. For example, Bitcoin \cite{nakamoto2008bitcoin} processes only $4.6$ transactions per second (TPS), and Ethereum \cite{wood2014ethereum} manages $14.3$ TPS. In contrast, traditional electronic payment systems like Visa support up to $47000$ TPS.

To address the scalability issue, zk-rollups have been proposed, allowing for frequent off-chain (Layer 2, L2) transactions and only recording aggregated compressed data on the public blockchain (Layer 1, L1). They employ zk-SNARKs (zero-knowledge succinct non-interactive argument of knowledge) to ensure the validity of aggregated compressed transactions \cite{gramoli2015rollup, buterin2018chain, Floersch2019Ethereum, dYdX}. Currently, the block size of the underlying L1 blockchain is approximately 170KB \cite{Average_Block_Size}, and zk-rollups can increase throughput to 2000 TPS \cite{HybridLayer2}. To further improve the TPS of zk-rollups, the Ethereum community is exploring new techniques. For instance, EIP-4844 (Proto-Danksharding) \cite{EIP4844}, proposes a temporary storage design known as ``blob'' transactions in Ethereum blocks. Blob transactions can provide an additional 2MB storage for zk-rollups, and the advanced version, Danksharding \cite{EthereumDanksharding}, supports $32$ MB storage. With the blob transactions, zk-rollups can support much higher TPS. 

Although zk-rollups are promising with Danksharding, they introduce their own challenges. The first issue is known as the \emph{data availability}, which requires all transactions to be accessible. Danksharding incurs two new problems that affect data availability. First, L1 nodes cannot fully ensure the data availability of blob due to the \emph{lazy validator problem}. Second, historical transactions will be lost permanently once temporary blobs are deleted (\emph{historical data problem}). Although L2 nodes can store complete transaction data \cite{Validium, Volition, zkPorter}, data off-chain solutions incur some other data availability issues, such as the \textit{data withholding attacks} \cite{Data_Withholding_Attacks} hindering the withdraws from users.\footnote{Precisely, solutions that store data completely off-chain are sidechain approaches. In this paper, we regard them as a variant of zk-rollups.}

The second challenge is \emph{decentralization} \cite{chu2018curses}. Since existing zk-rollups use zero-knowledge proofs to show that aggregated compressed transactions are generated properly, efficiently generating zk-SNARK proofs places prohibitively high hardware demands on L2 nodes. Additionally, zk-rollups need to preserve historical transaction data. Thus, L2 nodes require more bandwidth and storage to download and preserve historical transactions. As a result, the L2 network becomes less decentralized due to the high requirement of hardware, which leads to subsequent problems like potential maximal extractable value (MEV) attacks \cite{zhou2021just,qin2022quantifying}.

In this paper, we present a set of novel techniques to enhance data availability and decentralization of L2 networks. To address data availability, we first propose a ``proof of download'' technique where L2 nodes are required to download transaction data and update states before generating the next L2 block. Additionally, we introduce ``proof of existence'' to ensure historical transactions exist in L2. Compared with the data storage solutions such as proof of retrievability (PoR) \cite{anthoine2021dynamic}, proof of replication (PoRep) \cite{fisch2018poreps}, and proof of storage-time (PoSt) \cite{ateniese2020proof}, our ``proof of existence'' focuses on data availability instead of storage and is more efficient and suitable for L2. By further leveraging parallel processing and partial storage, these techniques improve data availability while alleviating the bandwidth and storage burdens on L2 nodes, which also partially address the decentralization issue.

To address the decentralization challenge, we design a novel ``role separation'' scheme for L2 nodes, which enables nodes with limited hardware resources to participate in L2 networks and perform the task of selecting preferred transactions. Furthermore, we propose an innovative ``proof of luck'' technique to prevent collusion among L2 nodes. As an additional benefit, our system demonstrates the potential to mitigate MEV attacks, which can be of independent interest to other MEV-resistant solutions. We conclude the main contributions of our paper as follows:

\begin{itemize}
\item  \emph{Ensuring data availability.} We propose a novel secure data off-chain design to guarantee data availability of zk-rollups (all participants can access the data). First, we introduce the ``hidden state'' technique for the proof of download, ensuring L2 nodes cannot generate new blocks without downloading historical transactions (\secref{sec: Proof of Download}). Second, to guarantee the long-term accessibility of transactions, we further propose a data availability challenging scheme. This strategy punishes nodes for maliciously deleting historical data and ensures a timely response for each query (\secref{sec: Proof of Existence}). Moreover, we enhance our designs with parallel processing and partial storage, which help to alleviate the bandwidth and storage burdens on L2 nodes. These optimizations improve efficiency while maintaining data availability.

\item \emph{Improving decentralization.} We present a novel protocol to enhance the decentralization of L2 networks. By separating L2 nodes into proposers for selecting transactions and builders for generating aggregated compressed transactions and corresponding proofs, we reduce the hardware requirements for participating in L2, thereby improving decentralization (\secref{sec: New Role Separation}). Additionally, with the help of our ``proof of luck'' and ``period separation'' techniques, we can effectively prevent collusion between proposers and builders (\secref{sec: Proof of Luck and Period Separation}), reinforcing the security of the system and resisting potential MEV attacks (\secref{sec: Against MEV Attacks}).

\item \emph{Implementation and evaluation.} We integrate the proposed approaches to implement a zk-rollup prototype and evaluate its performance and effectiveness (\secref{sec: Evaluation}). Through extensive experiments, we demonstrate the practicality and efficiency of our designs in terms of data availability, decentralization, and resistance to MEV attacks. Our results confirm the viability of our techniques and their potential for real-world adoption. 
\end{itemize}

\section{Related Work}
\noindent \textbf{Data availability}.
Several off-chain scaling solutions have been proposed to improve the throughput of on-chain transactions. Based on the way of ensuring data availability, these scaling solutions can be divided into rollups and sidechains.

Zk-rollup systems address data availability issues by including aggregated compressed transactions in the \texttt{CALLDATA} field of the L1 blockchain \cite{buterin2018chain,zkSync, Loopring}. Unfortunately, the throughput of these approaches is upper bound by the maximum block size and block generation rate of L1. To alleviate this problem, Danksharding provides additional space named ``blob'', which can be up to at most 32MB in each block. In order to reduce the bandwidth burden, L1 validators are only required to sample and download part of blob data to ensure data availability with the data availability sampling (DAS) technique \cite{Danksharding}. Though Danksharding is promising (supported by EIP-4844 \cite{EIP4844}), it incurs two problems impacting data availability. First, L1 nodes accept batches directly without sampling for better efficiency, known as the lazy validator problem. Second, since the blob is a temporary storage space, detailed transactions will be deleted after a period. How to preserve all historical data remains an unsolved issue.

Sidechain-based schemes include plasma \cite{poon2017plasma} and validium \cite{Validium, ZKFair, ImmutableX}. In these approaches, detailed L2 transactions are stored off-chain. Similar to zk-rollups, validium also use validity proofs such as zk-SNARK. The major difference between validium and zk-rollups is that validium does not store transaction data on the Ethereum Mainnet while zk-rollups offer on-chain data availability. Data availability of validium is ensured by the data availability managers who attest to the availability of data for off-chain transactions by signatures. The L1 verifier contract checks them before approving state updates. However, if data availability managers on the validium chain withhold off-chain state data from users, their funds will be frozen since they cannot compute the Merkle proof to prove the ownership of funds and execute withdrawals.

\noindent \textbf{Decentralization}.
Zk-rollups require L2 nodes to efficiently generate proofs to show the transactions are aggregated properly, which poses a heavy burden on L2 nodes and reduces the decentralization of L2. To address the decentralization problem, most of the existing approaches adopt new incentive schemes to attract more nodes \cite{Validium, Volition, zkPorter}, but fail to reduce the hardware requirement. Optimistic rollups \cite{Floersch2019Ethereum, Base, Mode} accept batches (L2 blocks) without the proof and punish nodes who submit invalid batches with fraud proofs, such as Optimism \cite{Optimism} and Arbitrum One \cite{ArbitrumOne}. Though these solutions can reduce computational requirements, users need to wait for a long withdrawal period (up to 20 minutes in existing designs) before spending their money, ensuring others have enough time to challenge batches with fraud proofs (this problem is completely avoided in zk-rollups). Besides, these approaches cannot address the MEV attack \cite{zhou2021just,qin2022quantifying}, which is a serious threat in L2 applications. 

\noindent \textbf{Data storage proof}.
Although our proposed ``proof of existence'' focuses on data availability (i.e. all nodes can access the data), it also partially guarantees historical data storage. We compared it with the existing data storage proof schemes that ensure data is stored honestly.

A successful proof of retrievability (PoR) \cite{anthoine2021dynamic} audit provides a strong guarantee of retrievability: if the server alters many blocks, this would be detected with high probability, whereas if only a few blocks are altered or deleted, then the error correction techniques ensure the file can still likely be recovered. Proofs of space (PoS) \cite{ateniese2014proofs, dziembowski2015proofs} require the prover to use the specified amount of memory to compute proofs. A proof of replication (PoRep) \cite{fisch2018poreps, Filecoin} builds on the two prior concepts of PoR and PoS, which proves that multiple copies of a data file are stored remotely. Proof of storage-time (PoSt) \cite{ateniese2020proof} is initially proposed in the Filecoin whitepaper, ensuring data is stored honestly over a period. Note that these schemes focus on data storage but cannot ensure data availability even when L2 nodes honestly store all transaction data. L2 nodes can withhold the required transaction data instead of responding to the challenger, which makes the transaction data stored but unavailable. Besides, PoR, PoRep, and PoSt sacrifice efficiency to achieve various functions such as multiple separate copies storage and retrievability, which violates the decentralization requirements. Thus, adopting these schemes directly to ensure data availability in our system is not suitable.

\begin{figure}
    \includegraphics[width=0.9\columnwidth]{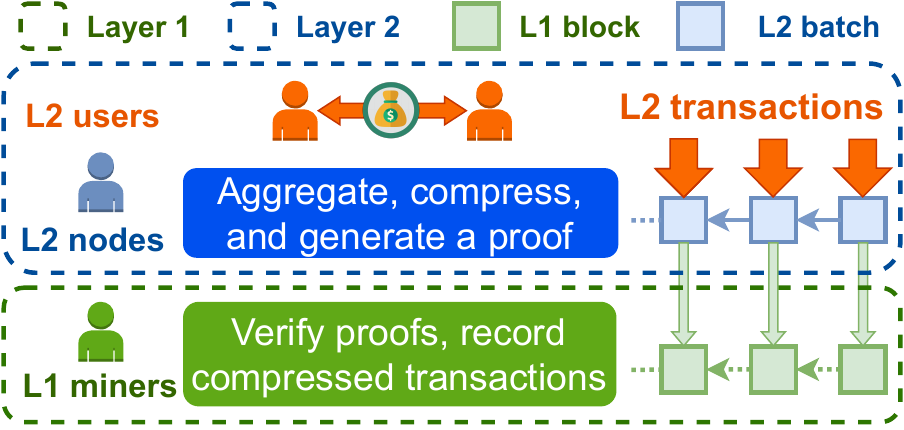}
    \caption{System model of L1 and L2. L2 nodes will compress and aggregate transactions, then send these (compressed) aggregated transactions to L1, together with a proof to show the accuracy of the aforementioned execution.}
    \label{system model}       
\end{figure}

\section{Preliminaries}
\subsection{Blockchain and Layer 2}
\label{sec: Blockchain and Layer 2}
A blockchain is a distributed ledger where transactions are recorded in blocks. In this paper, we refer to public blockchains like Ethereum as L1. L1 nodes are referred to as miners and they are responsible for storing the account states. The L2 network is a secondary framework built upon L1, aimed at improving the scalability of L1 by compressing and aggregating transactions. The blocks in the L2 network are referred to as batches. The L2 node needs to (1) select the transactions to be included in the new batch, (2) compress and aggregate the selected transactions, (3) generate a zk-SNARK proof to show that the (compressed) aggregated transactions are correctly computed from selected ones, (4) broadcast the newly created batch in the L2 networks and send the aggregated transactions to L1, and (5) update the local status when a new valid batch is received. The system model of L1 and L2 is illustrated in \figref{system model}.

In existing zk-rollup systems such as zkSync Era, the L2 network does not necessarily need to synchronize all batches to L1. This is primarily due to the limited number of existing L2 transactions, and synchronizing batches without any transactions would be meaningless. However, this leads to the problem of withdrawal delay, which is similar to optimistic rollup. Since we focus on scenarios with high TPS, we require each batch to be synchronized on L1 in our model.

As our main emphasis is on addressing data availability and decentralization challenges in L2, we do not specify the consensus algorithms used in L2, which implies our solutions are applicable to all consensus algorithms. In real-world scenarios, the most commonly used consensus algorithm for L2 is proof-of-stake (PoS).

\subsection{Notation}
Let $\lambda \in \mathbb{N}$ be a security parameter. $1^\lambda$ denotes its unary representation, and $\mathrm{negl}(\lambda)$ is a negligible function with respect to the security parameter $\lambda$. We denote the process of uniformly sampling an element $x$ from the set $\mathcal{S}$ by $x \overset{\$}{\leftarrow} \mathcal{S}$.

We use $\ZZ_p$ to denote the ring of integers modulo $p$ and $\GG$ to represent a cyclic group with order $p$. Consider two groups, $\GG$ and $\GG_T$. The bilinear mapping function is denoted by $e(\cdot) : \GG \times \GG \rightarrow \GG_T$, where $e(g^a, g^b) = e(g, g)^{ab} \in \GG_T$.

In our design, we make use of several hash functions: $H$, $H_1$, $H_2$, $H_3$, and $H_4$, where $H_1, H_2, H_3, H_4: \{0, 1\}^* \rightarrow \mathbb{Z}_p$. The input space of $H$ is $\{0, 1\}^*$, and we do not specify the output space in the general design. However, in the specific implementation, we specify $H: \{0, 1\}^* \rightarrow \GG$. In some cases, we also require the hash function to deal with multiple inputs $s_1, \cdots, s_m$. Unless otherwise specified, this can be implemented by processing the inputs as $s_1 \| \cdots \| s_m$.

\subsection{Building Block}
\subsubsection{KZG Commitment}
The KZG commitment scheme, introduced by Kate et al. \cite{kate2010constant}, is a constant-size polynomial commitment scheme that enables a prover to show the evaluations of a polynomial at specific points. Since the KZG supports batch verification \cite{boneh2021halo, boneh2020efficient, zhang2022polynomial} for the evaluation of multiple points and polynomials, we employ KZG as the polynomial commitment in our design to enhance efficiency. In our system, we utilize five functions of KZG: $\mathsf{Setup}$, $\mathsf{Commit}$, $\mathsf{Open}$, $\mathsf{Eval}$, and $\mathsf{VerifyEval}$.

\begin{itemize}
    \item $\mathsf{srs} \leftarrow \mathsf{Setup}(1^{\lambda}, D)$: Generate a group $\GG$ based on the security parameter $\lambda$. Sample a group element $g \overset{\$}{\leftarrow} \GG$ and a secret $\alpha \overset{\$}{\leftarrow} \mathbb{Z}_p$. Compute and return the public structured reference string $\mathsf{srs}=(g, g^{\alpha}, ..., g^{\alpha^{D}})$.
    
    \item $C \leftarrow \mathsf{Commit}(\mathsf{srs}, \phi(x))$: Given a $k$-degree polynomial $\phi(x)=\sum_{i=0}^{k}\phi_ix^i$ where $k \leq D$, compute and return its commitment $C=g^{\phi(\alpha)}=\prod_{i=0}^{k}(g^{\alpha^i})^{\phi_i}$.
    
    \item $b \leftarrow \mathsf{Open}(\mathsf{srs}, C, \phi(x))$: Verify that $C$ is the commitment of $\phi(x)$ by checking $C \overset{?}{=} \mathsf{Commit}(\mathsf{srs}, \phi(x))$. Return $b = 1$ if the equation holds; otherwise, $b = 0$.

    \item $(i, \phi(i), \pi_i) \leftarrow \mathsf{Eval}(\mathsf{srs}, \phi(x), i)$: Evaluate $\phi(x)$ at $i$, i.e., $\phi(i)$, and compute the corresponding proof $\pi_i=g^{\frac{\phi(\alpha)-\phi(i)}{\alpha-i}}$ based on $\mathsf{srs}$ in a similar manner to $\mathsf{Commit}$. Return $(i, \phi(i), \pi_i)$.

    \item  $b \leftarrow \mathsf{VerifyEval}(\mathsf{srs}, C, i, \phi(i), \pi_i)$: Verify that $\phi(i)$ is indeed the evaluation of the polynomial committed to $C$ at $i$ by checking $e(C/g^{\phi(i)}, g)$ $\overset{\text{?}}{=}e(\pi_i, g^\alpha/g^i)$. Return $b = 1$ if the equation holds; otherwise, $b = 0$.
\end{itemize}

In our design, we sometimes omit the input $\mathsf{srs}$ for simplicity. KZG commitment satisfies three properties: \textit{completeness}, \textit{polynomial binding}, and \textit{evaluation binding} \cite{kate2010constant}.

\subsubsection{ZK-SNARK}
Let $\mathcal{R}$ be a nondeterministic polynomial (NP) relation. For a pair $(x, w) \in \mathcal{R}$, we refer $x$ to the \emph{statement} and $w$ to the \emph{witness}. A zk-SNARK allows for the succinct proving of the relation $\mathcal{R}$ without revealing any information about the private input $w$. It consists of three polynomial-time algorithms: $\mathsf{Setup}$, $\mathsf{Prove}$, and $\mathsf{Verify}$, which are defined as follows.
\begin{itemize}
    \item $(\mathsf{pk}, \mathsf{vk}) \leftarrow \mathsf{Setup}(\mathcal{R}, \lambda)$: On inputs of the security parameter $\lambda$ and relation $R$, output the prover's proving key $\mathsf{pk}$, and the verifier's verification key $\mathsf{vk}$.  
    \item $\pi \leftarrow \mathsf{Prove}(\mathsf{pk}, x, w)$: Given the statement $x$ and the witness $w$, generate a succinct proof $\pi$ for the public relation $\mathcal{R}$ using the proving key $\mathsf{pk}$.   
    \item $b \leftarrow \mathsf{Verify}(\mathsf{vk}, x, \pi)$: Given the statement $x$ and the proof $\pi$, check whether the relation $\mathcal{R}$ holds using the verification key $\mathsf{vk}$. Output $b = 1$ if accepted and $b = 0$ otherwise.
\end{itemize}

Sometimes, we omit the inputs $\mathsf{pk}$ and $\mathsf{vk}$ for simplicity. ZK-SNARK satisfies four properties: \textit{succinctness}, \textit{completeness}, \textit{soundness}, and \textit{zero-knowledge} \cite{gabizon2019plonk}.

\subsubsection{EIP-4844}
The Ethereum improvement proposal 4844 (EIP-4844) \cite{EIP4844} aims to reduce the gas consumption of rollups by introducing a new transaction format called ``blob transaction''. The blob field in this format can contain a large amount of data that cannot be accessed during Ethereum virtual machine execution. However, the commitment to this data is accessible. By leveraging this design, blob transactions significantly lower the gas fees associated with data storage. The data contained in the blob field can be deleted after a predefined period, offering temporary but substantial scaling relief for rollups.

EIP-4844 is implemented in 2024. Once deployed, this new design will enable rollups to reduce $40$-$100$ times of gas fee compared to existing solutions that utilize \texttt{CALLDATA} storage mechanisms. 

\subsection{Security Assumption}
\label{sec: Security Assumption}
We make the following assumptions in our system design, which are consistent with the assumptions of other schemes \cite{TrustModels}.
\begin{itemize}
\item L1 is trusted and more than 50\% L2 nodes are honest \cite{Optimism, Polygonscan}.
As it is critical to allow all L2 nodes to accept the proposed proof of download (hidden state) in the newly generated batch, if most L2 nodes are dishonest, they can claim an invalid hidden state as a valid one. Additionally, when the L2 batch generation rate is higher than L1 block generation\footnote{In terms of block/batch propagation rate of L1 and L2, this paper mainly focuses on the case when two rates are the same. It is easy to extend our techniques to the scenario when the two rates are different.}, the validity of some batches entirely relies on the consensus of L2, which requires a majority of nodes to be honest.

\item In the historical data storage scenario, we do not consider collusion between different L2 nodes. The assumption is reasonable because this collusion does not damage data availability. For example, when node \textit{B} provides complete transaction data for node \textit{A}, \textit{A} can always respond to the challenger if it can quickly obtain the challenged data from \textit{B}. At this point, the challenger cannot distinguish between this scenario and the case when \textit{A} never deletes the transaction data. From the user's perspective, data availability still holds, though this situation may reduce the robustness of the entire system. Besides, our design can partially resist this attack. 

\item Malicious L2 nodes are profit-driven and can launch attacks for more profits, such as deleting historical transaction data to reduce storage costs.
\end{itemize}

\section{Observation and Problem Statement}

\noindent \textbf{Observation}.
We analyze the cost of L2 nodes, considering three main components: bandwidth cost, storage cost, and computation cost. With the introduction of blob transactions from Danksharding, Ethereum's TPS will improve significantly due to larger L1 blocks (from $100$KB to $32$MB), posing a challenge for L2 networks. Let's consider an example that a $24$MB batch need to be generated within $12$s. In this case, the node's bandwidth requirement is $10$MB/s, and it needs approximately $6$TB of storage per month. These bandwidth and storage costs amount to around $300$ and $200$ USD/month, respectively, representing a significant expense for L2 nodes. Besides, computation cost is also significant. For instance, a batch with $50$ transactions using a Groth16 circuit requires about $2^{25}$ constraints to generate a proof on $32$ vCPUs \cite{chen2023hyperplonk}. The cost of generating a single zk-SNARK proof for a batch of $50$ transactions is approximately $15.1$ USD and will increase significantly when including more transactions \cite{chen2023hyperplonk}.

Based on the above analysis, the primary issue contributing to data availability is the costly nature of bandwidth and storage. L2 nodes may be reluctant to download and store batches honestly due to the bandwidth and storage costs. On the other hand, the challenge of decentralization in L2 networks stems from the high costs associated with computation. This expense may discourage users from participating in L2 networks.

\noindent \textbf{Problem Statement}.
This paper addresses the challenges of data availability and decentralization in L2 networks. In terms of data availability, zk-rollups encounter issues such as lazy validator and historical data problems after adopting Danksharding. One potential solution for the lazy validator problem is to require all L2 nodes to \textit{download} the batch and update their states independently of L1, ensuring previous batches are available for other L2 nodes. Unfortunately, L2 nodes may lack the incentive to download due to the bandwidth costs. To address the historical data problem, the L2 network should \textit{maintain} all historical data and allow users to access it when queried. However, even if all historical data are downloaded, L2 nodes may delete them to save storage. Besides, malicious L2 nodes could withhold the challenged data from the user, even if they have stored the data.

For decentralization, since zk-rollups require L2 nodes to generate zk-SNARK proofs within a short period, we need to reduce the computation requirement of joining L2. Additionally, data availability requires L2 to download and store data honestly, which incurs a heavy burden on bandwidth and storage. In our system, we also need to reduce these costs in our design.

We summarize the problems addressed in this paper as follows.

\textbf{Problem 1}: How can we guarantee that L2 nodes download transaction data?

\textbf{Problem 2}: How to ensure historical data exists in L2 and is accessible to user?

\textbf{Problem 3}: How can we reduce the hardware requirements of L2 nodes for decentralization?

\noindent \emph{Remarks}.
This paper focuses on \textit{data availability} instead of \textit{data storage}.
There are two cases that can illustrate the difference between data availability and data storage. First, malicious L2 nodes store transaction data honestly but withhold it when a challenger requires it. In the view of the challenger, the required transaction data is unavailable, even if it is indeed stored in L2. In this case, data storage is ensured, but data availability is violated. Second, malicious L2 nodes can delete transaction data and re-download the deleted data from other nodes once it is queried. In the challenger's view, the requested transaction data is available, though it is deleted on this node most of time. In this case, data availability is ensured but data storage is violated.

\section{Data Availability}


\subsection{Proof of Download}
\label{sec: Proof of Download}

To address the lazy validator problem in existing zk-rollup solutions, we introduce the concept of ``\emph{proof of download}'', which requires that L2 nodes are obligated to verify data availability (download complete transactions) before generating the next L2 batch. If L2 nodes are too lazy to check data availability, they cannot generate the latest batch. Besides, L2 nodes must download previous transactions, partly addressing the historical data problem. We define proof of download as follows:

    
    

\begin{definition}[Proof of Download] 
\label{def: Proof of Download}
Proof of download is calculated based on the previous transactions and will be included in the header of the next L2 batch. 
Proof of download satisfies two properties: 1) honest L2 nodes who download previous transactions can compute proof of download for batch generation and verification (\textbf{completeness}), and 2) malicious L2 nodes without previous transactions cannot forge valid proof of download (\textbf{soundness}).
\end{definition}




\noindent \textbf{Hidden state}.
To achieve proof of download, we propose a new design called ``\emph{hidden state}'', which is an additional field in the L2 batch header. This hidden state is computed based on all transactions in the previous batch(es). Denote the transactions and hidden state in the $i$-th batch as $\mathrm{TX}_i$ and $\mathrm{hidden\_state}_i$, respectively. One straightforward design approach is to calculate the hidden state as $\mathrm{hidden\_state}_i = H(\mathrm{TX}_{i-1})$, where $H$ is a hash function. When generating the $i$-th batch, L2 nodes compute and include $\mathrm{hidden\_state}_i = H(\mathrm{TX}_{i-1})$ in the batch header. Once the $i$-th batch is published, other L2 nodes can verify whether $\mathrm{hidden\_state}_i$ is correctly computed. Since $\mathrm{hidden\_state}_i$ cannot be forged without knowing $\mathrm{TX}_{i-1}$, and the $\mathrm{hidden\_state}_i$ is not publicly available before an L2 node publishes it, all L2 nodes must download $\mathrm{TX}_{i-1}$ to participate in the generation of the $i$-th batch (We also consider the case that L2 nodes can share the hidden state with others. Although an L2 node can obtain an out-of-data hidden state from others, it cannot generate the latest batch based on the out-of-data hidden state, see more in \secref{sec: Discussion}). When synchronizing the aggregated transactions with L1, we also require the hidden state to be recorded on L1. However, the straightforward design described above incurs two issues: 
\begin{itemize}
    \item Increasing bandwidth. The current design is serial processing, which requires L2 nodes to download $\mathrm{TX}_{i-1}$ before calculating the $\mathrm{hidden\_state}_i$. This increases the L2 nodes' bandwidth requirements when TPS increases. 
    \item Unfair competition. In some L2 networks where nodes must compete for generating batches, the node successfully generating the previous batch has more time to generate the next batch since it has $\mathrm{TX}_{i-1}$ and can compute the $\mathrm{hidden\_state}_i$ directly.
\end{itemize} 



\noindent \textbf{Parallel processing}.
To address the mentioned issues, we set the hidden state of $\text{\em batch}_{i}$ as $\mathrm{hidden\_state}_i = H(\mathrm{TX}_{i-2})$. This enables parallel processing, allowing L2 nodes to download previous transactions and generate new batches simultaneously. This idea is based on the fact that an honest L2 node already has $\mathrm{TX}_1,..., \mathrm{TX}_{i-2}$ when generating $\text{\em batch}_{i}$. Therefore, it can use $\mathrm{TX}_{i-2}$ to compute $\mathrm{hidden\_state}_i$ and download $\mathrm{TX}_{i-1}$ simultaneously during the batch generation period. Note that L2 nodes also need to obtain the latest states before generating the next batch. This can be done by synchronizing with L1, which is much more efficient than deriving from $\mathrm{TX}_{i-1}$. 
With this design, both problems mentioned earlier can be addressed:



\begin{itemize}
    \item For the bandwidth issue, originally, L2 nodes are required to download $\mathrm{TX}_{i - 1}$ before generating $\text{\em batch}_{i}$ (serial processing). Both tasks must be completed within the batch generation period. However, in the new parallel processing design, L2 nodes can utilize the (almost) entire period to download $\mathrm{TX}_{i - 1}$. This modification effectively reduces the bandwidth requirement. Note that by reducing the bandwidth requirement, we enable more nodes to participate in the L2 network. This increased participation contributes to the decentralization of the L2 network. 
    \item For the unfair competition, the time difference is significantly reduced in the new design because L2 nodes only need to update their states from L1 using the aggregated compressed transactions, which are much smaller compared to the complete transactions. For instance, in the Loopring network, the aggregated transactions after compression are approximately $12$ Bytes, whereas complete transactions are about $310$ Bytes \cite{sguanci2021layer}.
\end{itemize}



\noindent \emph{Remarks}. 
We can also associate the hidden state with other historical transaction(s). For example, we can set (1) $\mathrm{hidden\_state}_i = H(\mathrm{TX}_{i-k})$ for some $k>2$; (2) $\mathrm{hidden\_state}_i = H(\mathrm{TX}_{i-2}, \mathrm{TX}_{i-k_1}, $ $\cdots, \mathrm{TX}_{i-k_n})$ where $k_1, \cdots, k_n$ are random challenges (generated from L2's consensus algorithm); or (3) $\mathrm{hidden\_state}_i = H(\mathrm{TX}_1, \cdots, $ $\mathrm{TX}_{i - 2})$. The first variate ensures that L2 nodes must download $\mathrm{TX}_{i-k}$ before generating the $i$-th batch. The bandwidth cost for this variation is similar to the $k = 2$ case since the downloading speed should be at least as fast as the batch generation rate. In the second variation, the L2 network tests whether L2 nodes process $\mathrm{TX}_{i-2}, \mathrm{TX}_{i-k_1}, \cdots, \mathrm{TX}_{i-k_n}$. Without these transactions, L2 nodes are unable to generate the $i$-th batch. For the last design, the hidden state guarantees that L2 nodes have processed all previous transactions. 
Note that the computation of the latest hidden states cannot be derived from the intermediate values of previous hidden states. This requires the intermediate values of $H(M)$ and $H(M')$ to be independent, even when $M \cap M' \neq \varnothing$.

\subsection{Proof of Existence}
\label{sec: Proof of Existence}
While the hidden state ensures that L2 nodes must download complete transactions, it does not prevent malicious nodes from deleting these data after generating batches and withholding them upon queries/challenges. One straightforward solution is to use the last variant of the hidden state, $\mathrm{hidden\_state}_i = H(\mathrm{TX}_1, \cdots, \mathrm{TX}_{i - 2})$. Unfortunately, it is inefficient for real-world designs as L2 nodes would need to read all historical transactions to compute the hidden state. To address these problems, we introduce the concept of ``\emph{proof of existence}'', which ensures all transactions exist and are accessible in L2. We define proof of existence as follows:

\begin{definition}[Proof of Existence]
\label{def: Proof of Existence}
Proof of existence is calculated based on the challenged data, which serves as proof to show the data exists and is accessible to the challenger.
Proof of existence satisfies two properties: 1) honest L2 nodes with the challenged data can compute valid proof of existence to allow the challenger to verify  (\textbf{completeness}), and 2) malicious L2 nodes without the data cannot forge valid proof of existence (\textbf{soundness}).
\end{definition}

\noindent \textbf{Data Availability Challenging}.
To ensure data existence and accessibility, we propose a challenging approach to punish nodes that fail to provide the expected transactions. Specifically, we require each node to deposit funds into an arbiter contract before joining the L2 network. Only nodes with a valid deposit are recognized as legitimate L2 nodes and are allowed to generate batches. Any participant, including users or other L2 nodes, can act as a challenger. When a challenger asks an L2 node for the detailed transactions of the $i$-th batch, the node is required to respond with $\mathrm{TX}_i$ (as proof of existence) to the arbiter contract within a predefined period. If the node fails to respond or provides an incorrect response, the arbiter contract transfers the payment to the challenger. This approach incentivizes L2 nodes to preserve all historical transactions in order to respond successfully to challenges. 


This approach ensures data existence by requiring each L2 node to store all historical transactions. However, with the increase of TPS, this requirement can impose a significant storage burden on L2 nodes. Recall data availability in L2 networks requires the entire historical transaction data can be recovered by the whole L2 network rather than an individual L2 node. This motivates us to adopt ``partial storage'' in more decentralized L2 networks. Each node is responsible for storing only a portion of the historical transactions to alleviate the storage burden while still ensuring the complete historical data is recoverable by all nodes with a high probability. Accordingly, some adjustments must be made to data availability challenging to accommodate this design.

\noindent \textbf{Partial storage}.
To enable each L2 node to store only a portion of the historical transactions, we adjust our hidden state design with the KZG commitment \cite{kate2010constant} for the $H$ function, allowing for the processing of more structured data. The fundamental idea is to divide the historical transactions $\mathrm{TX}_{i}$ into $k$ parts: $\mathrm{TX}_{i} = (\mathrm{TX}_{i}^{(0)}, \ldots, \mathrm{TX}_{i}^{(k - 1)})$. Each part is then processed with a hash function $H_1$, resulting in values $v_j = H_1(\mathrm{TX}_{i}^{(j)})$ for $j \in [0, k)$. Subsequently, an L2 node generates a $(k - 1)$-degree polynomial $\phi_{i}(X)$ such that $\phi_{i}(j) = v_j$. As a result, for transactions $\mathrm{TX}_{i}$ in the $i$-th batch, an L2 node only needs to store an arbitrary part $\mathrm{TX}_{i}^{(j)}$ and the corresponding proof $\pi_{i}^{(j)}$. It is important to note that this partial storage approach is particularly suitable for highly decentralized networks, where even with some failures of L2 nodes, the network can still recover the complete data.

 



\noindent \textbf{Efficient arbiter contract}.
\label{sec: Efficient arbiter contract}
We now show the logic of the arbiter contract. On receiving the challenge for the transactions in the $i$-th batch, the contract reads the $\mathrm{hidden\_state}_{i+2}$ from the $(i+2)$-th batch, which is a KZG commitment to the corresponding polynomial of $\mathrm{TX}_i$. Intuitively, the L2 node should respond with $\mathrm{TX}_{i}^{(j)}$ and $\pi_{i}^{(j)}$, allowing the contract to check with $\mathsf{KZG.VerifyEval}$, which proves both existence and accessibility. While this intuitive approach works when the size of $\mathrm{TX}_{i}^{(j)}$ is small, it can result in high transaction fees (gas) for L2 nodes, especially when $\mathrm{TX}_{i}^{(j)}$ grows larger in high throughput scenarios. Besides, in some cases,  users may simply wish to confirm the existence of their data rather than retrieve it. The straightforward approach will result in transaction fee waste for both users and L2 nodes when supporting frequent data existence checks (this method also incurs the DDoS problem which we discuss in the \secref{sec: Discussion}). 

To ensure efficiency, we construct two functions for data availability challenging: data-retrieval and data-checking. 

For data retrieval, we consider an L2 data storage market where users store their data on the L2 network. Users can invoke the data-retrieve function to access their own data. This request can be made to the entire L2 network (i.e., any L2 nodes can respond), or to a specific node (if the node indicates it possesses the expected data through offline communication). The function can be built with the straightforward logic mentioned above or with a fair-exchange protocol \cite{dziembowski2018fairswap, zkSync} for large data.

Additionally, we design an efficient data-checking function tailored for data existence checks. A naive idea is to let the L2 node being challenged respond with a value $v_j = H_1(\mathrm{TX}_i^{(j)})$ and a zk-SNARK proof $\pi_v$ to show $v_j$ is properly generated. However, it incurs a serious problem as the L2 node can store $v_j$ and $\pi_v$ instead of $\mathrm{TX}_{i}^{(j)}$, which does not satisfy the soundness of ``proof of existence''. The underlying problem lies in the fact that the proof $\pi_v$ remains valid for all challenges, whereas we expect it to be valid solely for the current one. 

To address this problem, we bind $\pi_v$ to a random challenge from the challenger. Specifically, when challenging the $i$-th batch, the challenger also sends a challenge $c$ to the arbiter contract. The L2 node should retrieve $c$ from the contract and calculate $r = H_2(c, \mathrm{TX}_{i}^{(j)})$ where $H_2$ is a hash function. The zk-SNARK proof shows both $v_j$ and $r$ are correctly computed, i.e., $\pi_v = \mathsf{SNARK.Prove}$ $((c, v_j, r), (\mathrm{TX}_i^{(j)}))$ which indicates the following relation
\begin{align}
    v_j = H_1(\mathrm{TX}_i^{(j)}) \quad \wedge \quad r = H_2(c, \mathrm{TX}_{i}^{(j)}).
\nonumber
\end{align}
Consequently, the arbiter contract checks $v_j$ using the zk-SNARK verification $\mathsf{SNARK.Verify}((c, v_j, r), \pi_v) \overset{?}{=} 1$ and the KZG verification for evaluation $\mathsf{KZG.VerifyEval}(\mathrm{hidden\_state}_{i+2}, j, v_j, \pi_{i}^{(j)}) \overset{?}{=} 1$. Since the size of $v_j$ and $\pi_{i}^{(j)}$ remain constant and both $\mathsf{SNARK.Verify}$ and $\mathsf{KZG.VerifyEval}$ operations require constant time, the gas fee will not increase even when $\mathrm{TX}_{i}^{(j)}$ is large.



\section{Decentralization}

\subsection{New Role Separation}
\label{sec: New Role Separation}
The current high computation requirement for zk-rollup networks poses a significant challenge to decentralization. L2 nodes need to generate a zk-SNARK proof to show each batch is valid within a short period. Prior to presenting our solution to this issue, we first recall the procedures of batch generation of L2 nodes (\secref{sec: Blockchain and Layer 2}). Among these tasks, the first one of transaction selection requires less computational power and is considered the ``brain work'' of zk-rollups. This insight forms the basis of our proposer/builder separation (PBS) technique,\footnote{Compared with the PBS in Danksharding \cite{PBS}, which works on L1, our L2 PBS is relatively simpler. Instead of directly adopting the PBS in Danksharding, we simplified PBS specifically tailored for L2 networks.} which separates the roles of L2 nodes into \emph{proposers} and \emph{builders}. 

\begin{itemize}
    \item \textit{Proposers} are L2 nodes with limited hardware resources. Their primary responsibility is to select the transactions that will be included in the next batch. Since the tasks of proposers do not impose much hardware requirement, anyone can participate as a proposer, allowing for a broader pool of participants and improving the decentralization of L2 networks.
    \item \textit{Builders} are L2 nodes equipped with more powerful hardware. They take the selected transactions from proposers and perform computationally intensive tasks such as transaction aggregation, zk-SNARK proof generation, batch storage, and broadcasting.
\end{itemize}
In our PBS design, we also utilize the temporal storage of blob transactions on L1 as a public board to facilitate coordination between proposers and builders. The basic workflow of our PBS scheme is illustrated in \figref{fig8}.

\begin{figure}
\centering
\includegraphics[width=0.9\columnwidth]{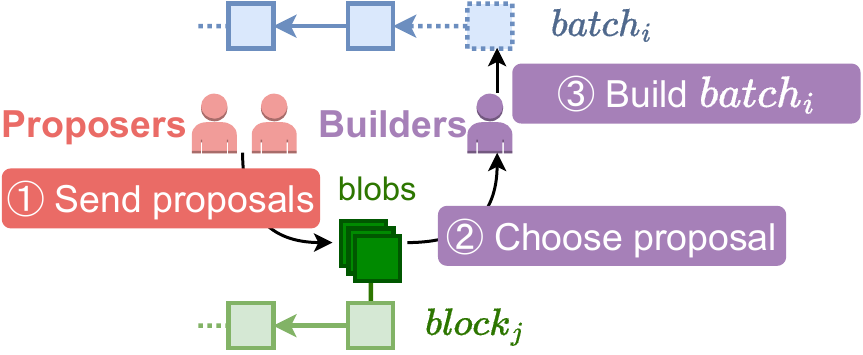}
	\caption{The role separation (PBS) in L2.}
	\label{fig8}
\end{figure}

\emph{Step 1}: Each proposer selects a set of transactions, denoted as $\{\mathsf{tx}_1, \cdots ,\mathsf{tx}_m\}$, from the transaction pool. It then computes a proposal, which includes $\{H_3(\mathsf{tx}_1), \cdots, H_3(\mathsf{tx}_m)\}$, where $H_3$ is a hash function. The proposer publishes the proposal on L1 within blob transactions.

\emph{Step 2}: To generate new batches, a builder selects a preferred proposal from L1 blocks and downloads the corresponding transactions from the transaction pool.

\emph{Step 3}: The builder performs the remaining tasks to generate the new batch, which includes compressing and aggregating transactions, generating the zk-SNARK proof, and broadcasting the newly created batch.

\subsection{Proof of Luck and Period Separation}
\label{sec: Proof of Luck and Period Separation}


In our role separation scheme, it is essential to prevent collusion between builders and proposers to maintain decentralization within the L2 network. If a builder with high computational power or stake consistently selects proposals from a specific proposer, the L2 network will be in centralized control. To address this concern and ensure that builders choose proposers randomly, we introduce a technique called ``\emph{proof of luck}'', which serves as a plug-in module within the underlying consensus algorithm:

\begin{itemize}
    \item Proposer $j$ is assigned with a unique ID denoted as $\mathrm{Proposer}_j$, which is signed by a recognized party upon joining the L2 network.
    
    \item To generate $\text{\em batch}_{i}$, builders need to compute a ``\textit{lucky number}'' based on the header of the previous batch, $\mathrm{luck}_{i} = H_4(\text{\em batch}_{i-1}.\mathrm{head})$, where $H_4$ is a hash function. 
    
    \item The probability of a builder generating the $i$-th batch is (partially) determined by the difference between the chosen proposer and the lucky number, such as $|\mathrm{Proposer}_j - \mathrm{luck}_{i}|$.
\end{itemize}
For example, in PoS-based networks, the builders are required to find an appropriate $r$ which is included in the batch header such that $\mathrm{Hash}(\text{\em batch}_i.\mathrm{head}(r)) < D(|\mathrm{Proposer}_j - \mathrm{luck}_{i}|)$, where $D$ is a monotone decreasing function to control the difficulty.\footnote{In the case of PoW-based networks, the difficulty can be directly set to $D = f(|\mathrm{Proposer}_j - \mathrm{luck}_{i}|)$ where $f$ is a monotone increasing function.} The probability of finding a proper $r$ within a period decreases if a builder includes a proposal from the proposer whose ID deviates from $\mathrm{luck}_{i}$.
Therefore, a malicious builder cannot collude with a proposer beforehand and increase their chances of success simultaneously. 

\noindent \emph{Remarks}. Our scheme also has flexibility since builders are advised to select the proposal whose proposer ID is close to the lucky number instead of choosing compulsorily a specific proposal. If there are two proposals whose proposer IDs are both close to the lucky number, honest builders can select the second proposal, which has more gas reward even if the first proposer ID is closer to the lucky number. 

However, this straightforward design only addresses \textit{premeditated} collusion. It does not prevent collusion \textit{after} the previous batch is published. Since the lucky number is determined by the previous batch, a builder could potentially engage in negotiations with proposers whose IDs are close to the $\mathrm{luck}_{i}$ once it is settled. To mitigate this issue, we propose the \emph{period separation} scheme.

\begin{figure}
\includegraphics[width=1\columnwidth]{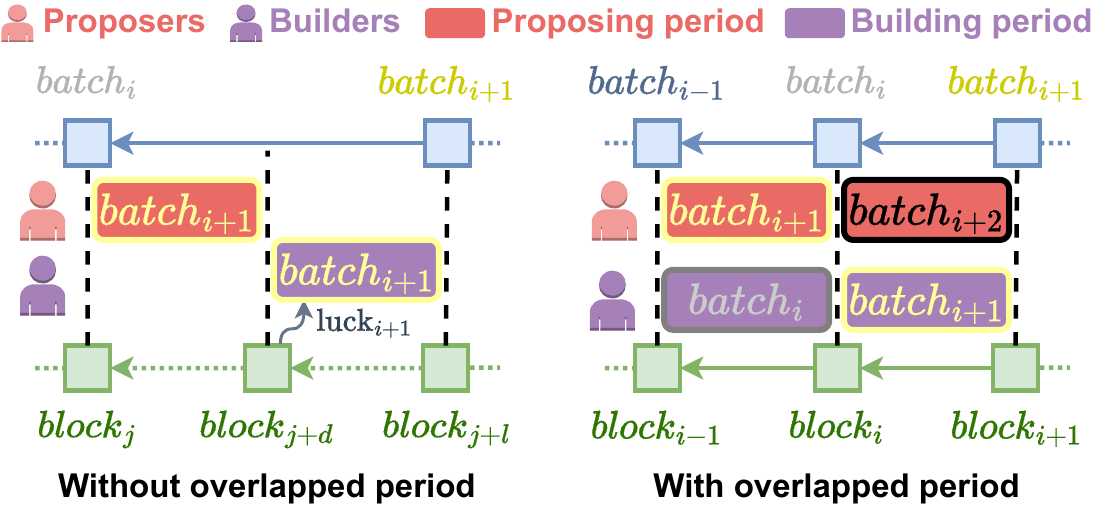}
	\caption{Period separation. The lucky number is set based on the last block in the proposing period.}
	\label{fig9}
\end{figure}



\noindent \textbf{Period separation}.
Suppose the batch generation period in L2 is $l$ times longer than the block generation period in L1. $\text{\em batch}_{i}$ is synchronized on $\text{\em block}_j$ of L1 and $\text{\em batch}_{i + 1}$ is (expected to be) synchronized on $\text{\em block}_{j + l}$. In our period separation scheme, we divide the batch generation period into two based on a predefined value $d$: the proposing period from L1 $\text{\em block}_j$ to $\text{\em block}_{j+d}$, and the building period from L1 $\text{\em block}_{j+d}$ to $\text{\em block}_{j + l}$, as illustrated in \figref{fig9} (without overlapped period). The detailed procedure is as follows:

\begin{itemize}
    \item \emph{Proposing period.} Proposers create their proposals and send them to L1. These proposals are recorded in blob transactions within blocks $\text{\em block}_j$ to $\text{\em block}_{j+d}$. Only proposals generated within this period are considered valid.
    \item \emph{Building period.} A builder computes the lucky number $\mathrm{luck}_{i + 1} $ $= H_4(\text{\em block}_{j + d}.\mathrm{head})$. Then, it selects a proposal created during the proposing period to build the new L2 batch. Notably, the lucky number is now computed based on $\text{\em block}_{j + d}$ rather than $\text{\em batch}_{i}$. This change ensures that proposers cannot submit new proposals once the lucky number has been determined, preventing any attempt to manipulate the selection process. 
\end{itemize}
However, this approach is only suitable when the batch generation time is (at least twice) longer than the block generation time in L1, i.e., $l \geq 2$. This condition reduces the upper bound of TPS since the theoretical batch generation rate can be the same as (or higher than) the block generation rate in L1. To address this limitation, we propose an \emph{overlapped} period separation scheme.

\noindent \textbf{Overlapped period separation}.
To improve the batch generation rate, we leverage parallel processing by allowing a period to serve as both the proposing period and the building period for different batches, i.e., the period from L1 $\text{\em block}_i$ to $\text{\em block}_{i+1}$ serves as the proposing period for $\text{\em batch}_{i + 1}$ and the building period for $\text{\em batch}_{i}$. The blob field of $\text{\em block}_{i}$ contains the proposals for $\text{\em batch}_{i+1}$. This design is illustrated in \figref{fig9} (with overlapped period) and the detailed procedure is described as follows: 

\begin{itemize}
    \item After the $\text{\em batch}_{i - 1}$ is synchronized on L1 $\text{\em block}_{i-1}$, proposers send their proposals for $\text{\em batch}_{i + 1}$, which will be recorded in the blob field of $\text{\em block}_i$. 
    \item Meanwhile, builders compute $\mathrm{luck}_i = H_4(\text{\em block}_{i-1}.\mathrm{header})$ and select proposals in $\text{\em block}_{i-1}$ to construct $\text{\em batch}_i$.
    \item Once the $\text{\em batch}_i$ is synchronized on $\text{\em block}_i$, proposers and builders repeat the aforementioned operations for $\text{\em batch}_{i+1}$.
\end{itemize}
In our design, it is crucial to ensure that a proposal originates from a legitimate proposer within the designated proposing period. We address this through two steps. First, the legitimacy of the proposers is verified by L1 miners, who check whether the proposers are signed by the recognized party (via a smart contract). Second, once a batch is published, other builders within the network validate that the corresponding proposal indeed originates from the correct proposing period and has not been tampered with.



\subsection{Against MEV Attacks (Bonus)}
\label{sec: Against MEV Attacks}
One bonus feature of our proposed system is its resilience against MEV attacks, which have been a significant concern in blockchain networks. MEV attacks involve L2 nodes taking advantage of profitable transactions by front-running, back-running, or sandwiching them with MEV transactions to maximize their profits at the expense of other users.

In our system, the role separation technique (PBS) ensures that the tasks of selecting transactions and building batches are performed by different L2 nodes. When building a batch, since all transactions are chosen from the selected proposal, builders cannot insert additional transactions but will decide the order of transactions. Furthermore, the collaboration between builders and proposers is avoided by proof of luck and period separation. This means that proposers cannot predict the order in which the MEV and target transactions will be placed within the batch. While it is possible to execute MEV attacks when an attacker instantiates many proposers to increase the chance of having a proposer ID close to the lucky number, the cost of MEV attacks also increases significantly, which makes the attack not profitable. First, proposers need to register with a recognized party (require some cost). Second, proposers need to pay gas fees for sending proposals to L1. Therefore, our design enhances the security and fairness of transaction execution for all participants.


\section{System Design} 
\label{sec: System Design}
\subsection{Initialization Phase}
Before activating the L2 network, two smart contracts need to be deployed: a smart contract to verify the validity of batches and an arbiter contract for data availability challenges. To set
up the network, a predefined value $D$ is determined, representing the maximum number of data parts in partial storage. The network runs the $\mathsf{srs} = \mathsf{KZG.Setup}(1^{\lambda}, D)$ to generate the public parameters for the KZG scheme, and \textsf{Setup} of zk-SNARK to get the proving key \textsf{sk} and the verification key \textsf{vk}. These setups can be carried out by a trusted third party or through multi-party computation.

L2 nodes are divided into two roles: proposers and builders. To join the L2 network, a proposer must register with a recognized party to obtain a unique proposer ID, and a builder is required to deposit funds into the $\adv\cdv$ contract. The total number of builders in the L2 network determines the value of $k$, representing the number of data parts in partial storage. 

\subsection{Batch Generation and Verification}
When generating the $i$-th batch, the proposer $\pdv$ and builder $\bdv$ perform the following tasks.

In the period from $\text{\em block}_{i-2}$ to $\text{\em block}_{i-1}$:
\begin{itemize}
    \item $\pdv$: \emph{(Build a proposal)} Select a set of transactions $\mathsf{tx}_1, \cdots,\mathsf{tx}_m$ from the transaction pool and construct a proposal that includes $H_3(\mathsf{tx}_1), \cdots, H_3(\mathsf{tx}_m)$. Send the proposal in a blob transaction to L1. The proposal is expected to be included in $\text{\em block}_{i-1}$ of L1.
\end{itemize}

In the period from $\text{\em block}_{i-1}$ to $\text{\em block}_{i}$:
\begin{itemize}
    \item $\bdv$: \emph{(Download transactions and update states)} Download the transactions $\mathrm{TX}_{i - 1}$ from the previous batch $\text{\em batch}_{i - 1}$ in L2. Update the latest state from $\text{\em block}_{i - 1}$ in the L1 network.
    \item $\bdv$: \emph{(Compute hidden state)} Divide transactions $\mathrm{TX}_{i-2}$ into $k$ parts, denoted as $\mathrm{TX}_{i-2} = (\mathrm{TX}_{i-2}^{(0)}, \cdots, \mathrm{TX}_{i-2}^{(k-1)})$. For each $j \in [0,k)$, compute $v_j = H_1(\mathrm{TX}_{i-2}^{(j)})$. Interpolate to derive a $(k-1)$-degree polynomial $\phi_{i-2}(x)$ based on $k$ points $((0, v_0), \cdots, (k-1, v_{k-1}))$. Set $\mathrm{hidden\_state}_{i}$ to $\mathsf{KZG.Commit}(\phi_{i-2}(x))$.
    \item $\bdv$: \emph{(Store partial data)} Randomly select a part $\mathrm{TX}_{i-2}^{(j)}$ and removes the other parts. Run $(j, \phi_{i-2}(j), \pi_i^{(j)}) = \mathsf{KZG.Eval}(\phi_{i-2}(x), j)$ to obtain $\pi_i^{(j)}$ for data availability challenges. Store the tuple $T_i = (j, \mathrm{TX}_{i-2}^{(j)}, \pi_i^{(j)})$.
    \item $\bdv$: \emph{(Select proposal)} Compute $\mathrm{luck}_i = H_4(\text{\em block}_{i-1}.\mathrm{header})$. Collect a preferred proposal from $\text{\em block}_{i-1}$ whose proposer ID is close to $\mathrm{luck}_i$. Generate a membership proof $\pi_{\mathsf{p}}$ to show the proposal is from $\text{\em block}_{i-1}$.
    \item $\bdv$: \emph{(Build batch)} Select transactions $\mathsf{tx}_1, \cdots, \mathsf{tx}_m$ based on the proposal and decide the order of them. Compress and aggregate transactions to construct $\text{\em batch}_{i}$. Generate a zk-SNARK proof $\pi_i$ to demonstrate that $\text{\em batch}_{i}$ is properly computed.
    \item $\bdv$: \emph{(Publish new batch)} Publish $\text{\em batch}_{i}$ on the L2 network. Send the compressed aggregated transactions, $\pi_i$, $\pi_{\mathsf{p}}$, the selected proposal, and $\mathrm{hidden\_state}_{i}$ to the smart contract on the L1 network.
\end{itemize}

To verify the validity of $\text{\em batch}_{i}$, other builders $\bdv$ and the smart contract $\sdv\cdv$ perform the following tasks.

\begin{itemize}
    \item $\bdv$: Check the $\mathrm{hidden\_state_{i}}$ sent to L1 is same as the $\mathrm{hidden\_state_{i}}$ in $\text{\em batch}_{i}$. Verify $\mathsf{KZG.Open}(\mathrm{hidden\_state}_{i}, \phi_{i-2}(x)) \overset{?}{=} 1$, where $\phi_{i-2}(x)$ is the interpolated polynomial based on $\mathrm{TX}_{i - 2}$. Perform additional checks if necessary. Sign a note for ``$\text{\em batch}_{i}$ is valid'' if all checks pass. Send the note to $\sdv\cdv$ on L1.
    \item $\sdv\cdv$: Verify that the proposal is from $\text{\em block}_{i-1}$ with $\pi_{\mathsf{p}}$ and is originated from a valid proposer. Check the compressed and aggregated transactions are properly computed from the proposal with $\pi_i$. Once all the necessary checks pass and enough valid notes are received, regard $\text{\em batch}_{i}$ as valid and record $\mathrm{hidden\_state}_{i}$. Proceed to update the states accordingly.
\end{itemize}

\subsection{Data Availability Challenging}
In our design, we allow anyone to act as a challenger $\cdv$ to challenge a builder $\bdv$ regarding the storage of historical transactions. This is done through the arbiter contract $\adv\cdv$ with the following procedures.

\begin{itemize}
    \item $\cdv$: Generate a random challenge request $\mathsf{req} = \mathsf{PoE.Challenge}(i)$ to challenge the transactions in the $i$-th batch. Send $\mathsf{req}$ to $\adv\cdv$.
    \item $\bdv$: Collect $\mathsf{req}$ from $\adv\cdv$. Retrieve $T_i = (j, \mathrm{TX}_{i}^{(j)}, \pi_i^{(j)})$ from the local storage. Run $\pi = \mathsf{PoE.Response}(\mathsf{pk_E}, \mathsf{req}, T_i)$ to derive PoE proof $\pi_\mathsf{E}$ as the response. Send the tuple $\pi_\mathsf{E}$ to $\adv\cdv$.
    \item $\adv\cdv$: Collect $\mathrm{hidden\_state}_{i+2}$ from L1. Verify whether $\mathrm{TX}_{i}^{(j)}$ is a valid part of $\mathrm{TX}_{i}$ with
    $\mathsf{PoE.Verify}(\mathsf{vk_E}, \mathsf{req}, \pi_\mathsf{E}, \mathrm{hidden\_state}_{i+2}) \overset{?}{=} 1$. If the check fails or $\bdv$ does not respond within a predefined period, transfer the deposit of $\bdv$ to $\cdv$.
\end{itemize}

\section{Security Analysis}
\label{sec: Security Analysis}
We conduct a formal security analysis of the proposed schemes. First, we demonstrate that honest builders can successfully finish the following tasks: hidden state computation, new batch generation, and response to data availability challenges. 

\begin{theorem}[Completeness]
\label{the: Completeness}
When generating $i$-th batch, an honest builder $\bdv$ can compute a valid proof of download (proof of download completeness in \defref{def: Proof of Download}). If $\bdv$ selects the correct proposal from $\text{\em block}_{i-1}$, the probability of finding a proper $r$ is maximized when building new batches. If $\bdv$ is challenged with the $i$-th batch, it can respond to the arbiter contract correctly within a period (proof of existence completeness in \defref{def: Proof of Existence}).
\end{theorem}

\begin{proof}
We first argue the completeness of proof of download. When $\bdv$ honestly downloads $\mathrm{TX}_{i - 2}$, it can compute $\phi_{i - 2}(x)$ and derive $\mathrm{hidden\_state_{i}}=\mathsf{KZG.Commit}(\phi_{i - 2}(x))$. When other honest builders receive the newly published $i$-th batch, they can compute the $\phi_{i-2}(x)$ based on $\mathrm{TX}_{i - 2}$ in their local storage and verify $\mathrm{hidden\_state_{i}}$ with $\mathsf{KZG.Open}(\mathrm{hidden\_state_{i}}, \phi_{i-2}(x)) \overset{?}{=} 1$. Since the $\mathrm{hidden\_state_{i}}$ is computed honestly by $\bdv$, the algorithm outputs $1$. This means others can accept the legal hidden state. The completeness of proof of download follows directly from the correctness of the KZG polynomial commitment scheme. 

Next, we discuss the probability of finding $r$ during the batch generation process, under the condition $\mathrm{Hash}(\text{\em batch}_i.\mathrm{head}(r)) < D(|\mathrm{Proposer}_j - \mathrm{luck}_{i}|)$. After $\bdv$ computing the luck number $\mathrm{luck}_{i}$, it will select a suitable proposal that minimizes $|\mathrm{Proposer}_j - \mathrm{luck}_{i}|$. Since $D$ is a monotone decreasing function, $D(|\mathrm{Proposer}_j - \mathrm{luck}_{i}|)$ will be maximized, thus indicating the broadest valid range. Consequently, an honest builder can have the highest chance of finding an appropriate $r$ within the batch generation period.

Finally, we argue the response to data availability challenges. When $\bdv$ is challenged with $c$ and the $i$-th batch, it can retrieve $(j, \mathrm{TX}_{i}^{(j)}, \pi_i^{(j)})$ from the local storage, calculate $v_j = H_1(\mathrm{TX}_i^{(j)})$ and $r = H_2(c, \mathrm{TX}_{i}^{(j)})$, and subsequently generate $\pi_v$. The tuple $(j, v_j, \pi_i^{(j)}, r, \pi_v)$ will be a valid response since $\mathsf{KZG.VerifyEval}$ $(\mathrm{hidden\_state}_{i+2}, j, v_j, \pi_{i}^{(j)}) = 1$ and $\mathsf{Verify}(\mathsf{pk}, (c, v_j, r), \pi_v) = 1$, following the correctness of the KZG polynomial commitment scheme and completeness of zk-SNARK, respectively. 
\end{proof}

Second, we show that if the malicious builder $\bdv$ does not download previous transactions $\mathrm{TX}_{i - 2}$, it cannot compute a valid hidden state $\mathrm{hidden\_state_{i}}$ (proof of download soundness in \defref{def: Proof of Download}).

\begin{theorem}[Proof of Download Soundness]
\label{the: Hidden State Soundness}
Consider a malicious builder $\bdv$ that forges a proof of download $\mathrm{hidden\_state_{i}}'$ without downloading $\mathrm{TX}_{i - 2}$. The probability of $\mathrm{hidden\_state_{i}}'$ being accepted by other builders is negligible.
\end{theorem}
\begin{proof}
Suppose $\bdv$ does not possess $\mathrm{TX}_{i - 2}$ when constructing $\text{\em batch}_i$. Without $\mathrm{TX}_{i - 2}$, $\bdv$ cannot compute $\phi_{i - 2}(x)$ and $\mathsf{KZG.Commit}\\(\phi_{i - 2}(x))$. Therefore, $\bdv$ has to forge a $\mathrm{hidden\_state_{i}}'$ for the hidden state field. With all but negligible probability, $\mathrm{hidden\_state_{i}}' \neq \mathrm{hidden\_state_{i}}$ holds when the output space of $\mathsf{KZG.Commit}$ is sufficiently large, where $\mathrm{hidden\_state_{i}} = \mathsf{KZG.Commit}(\phi_{i - 2}(x))$. Therefore, due to the polynomial binding property of the KZG polynomial commitment, $\mathsf{KZG.Open}(\mathrm{hidden\_state_{i}}', \phi_{i - 2}(x))$ will output $0$. This implies that a forged proof of download will not pass the verification of other builders. Hence, the probability of $\bdv$ successfully forging a valid proof of download is negligible.
\end{proof}

Thirdly, we demonstrate that if the malicious builder $\bdv$ deletes $\mathrm{TX}_{i}^{(j)}$, the arbiter contract will transfer its deposit to the challenger as a penalty upon receiving a challenge $c$ for the $i$-th batch (proof of existence soundness in \defref{def: Proof of Existence}).

\begin{theorem}[Proof of Existence Soundness]
\label{the: Response Unforgeability}
Consider a malicious builder $\bdv$ who deletes all $\mathrm{TX}_{i}^{(j)}$'s for $\mathrm{TX}_{i}$. Upon receiving a challenge $(i, c)$, $\bdv$ cannot generate a valid response to pass the verification of the arbiter contract.
\end{theorem}
\begin{proof}
Suppose $\bdv$ does not possess any $\mathrm{TX}_{i}^{(j)}$ for all $j$'s when receiving a challenge $(i, c)$. $\bdv$ needs to construct $v'_j$ and $\pi_i'^{(j)}$ for some $j$ as a part of the response. Let $v_j$ and $\pi_i^{(j)}$ be the correct values that satisfy $\mathsf{KZG.VerifyEval}(\mathrm{hidden\_state}_{i+2}, j, v_j, \pi_i^{(j)}) = 1$. If $v'_j = v_j$ and $\pi_i'^{(j)} = \pi_i^{(j)}$, without $\mathrm{TX}_{i}^{(j)}$, $\bdv$ will be unable to provide a valid proof $\pi_v$ to show $v_j = H_1(\mathrm{TX}_{i}^{(j)})$ and $r = H_2(c, \mathrm{TX}_{i}^{(j)})$, according to the soundness of of zk-SNARK. If $v'_j \neq v_j$ or $\pi_i'^{(j)} \neq \pi_i^{(j)}$, $\mathsf{KZG.VerifyEval}(\mathrm{hidden\_state}_{i+2}, j, v'_j, \pi_i'^{(j)}) = 0$ following the evaluation binding property of the KZG polynomial commitment. Consequently, with all but negligible probability, the forged response cannot pass the verification of the arbiter contract. The deposit will be transferred to the challenger.
\end{proof}

Lastly, we show that when a malicious builder $\bdv$ colludes with a specific proposer, the probability of generating a new batch within a period decreases.

\begin{theorem}[Collusion Resistance]
\label{the: Adjustable  Difficulty}
Consider a malicious builder $\bdv$ that colludes with $\mathrm{Proposer}_{k}$. $\bdv$ selects the proposal from $\mathrm{Proposer}_{k}$ when building the $i$-th batch. The probability of finding an $r$ that satisfies $\mathrm{Hash}(\text{\em batch}_i.\mathrm{head}(r)) < D(|\mathrm{Proposer}_k - \mathrm{luck}_{i}|)$ is lower than the probability of $\mathrm{Hash}(\text{\em batch}_i.\mathrm{head}(r)) < D(|\mathrm{Proposer}_j - \mathrm{luck}_{i}|)$, where $\mathrm{Proposer}_{j}$ represents the proposer whose proposal is selected by honest builders.
\end{theorem}
\begin{proof}
When generating $\text{\em batch}_{i}$, all valid proposals reside in $\text{\em block}_{i - 1}$. As $\mathrm{luck}_i = H_4(\text{\em block}_{i - 1}.\mathrm{header})$, $\mathrm{Proposer}_k$ cannot propose a valid proposal for $\text{\em batch}_{i}$ after $\mathrm{luck}_i$ is settle. For honest builders, proposals are chosen to minimize $|\mathrm{Proposer}_j - \mathrm{luck}_{i}|$. Thus, given a large number of proposals, we can state that $|\mathrm{Proposer}_j - \mathrm{luck}_{i}| < |\mathrm{Proposer}_k - \mathrm{luck}_{i}|$ holds with overwhelming probability. Since $D$ is a monotone decreasing function, $D(|\mathrm{Proposer}_j - \mathrm{luck}_{i}|)$ is larger than $D(|\mathrm{Proposer}_k - \mathrm{luck}_{i}|)$. Consequently, the probability of an $r$ such that $\mathrm{Hash}(\text{\em batch}_i.\mathrm{head}(r)) < D(|\mathrm{Proposer}_k - \mathrm{luck}_{i}|)$ is lower than the probability of satisfying $\mathrm{Hash}(\text{\em batch}_i.\mathrm{head}(r)) < D(|\mathrm{Proposer}_j - \mathrm{luck}_{i}|)$.
\end{proof}

\begin{figure}
	\centering
	\begin{minipage}{0.49\linewidth}
		\centering
		\includegraphics[width=1\textwidth]{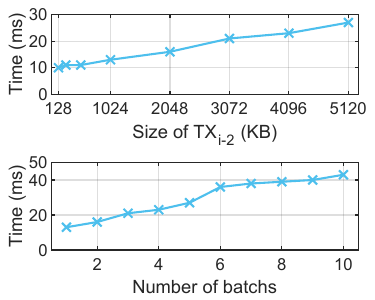}
		\caption{Hidden state generation time.}
		\label{hs time}     
	\end{minipage}
	\begin{minipage}{0.49\linewidth}
		\centering
		\includegraphics[width=1\textwidth]{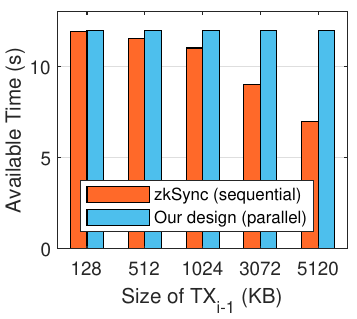}       
		\caption{Available time for building batches.}
		\label{ava time}     
	\end{minipage}
\end{figure}

\section{Evaluation}
\label{sec: Evaluation}
We implement a prototype of our proposed system to demonstrate its feasibility and evaluate the performance of our techniques\footnote{Our system prototype: https://github.com/i1Il1I1ll1ill1/Rollup.}. Our prototype runs on the Ethereum Goerli testnet. For the arbiter contract, we utilize the Circom library provided by iden3, which is implemented using the Rust language and employs the BN128 elliptic curve for pairing and group operations. Snark.js is used for the setup and generation of zero-knowledge proofs. We use Poseidon as the hash function for $H_1$, $H_2$, $H_3$, and $H_4$, which is a zk-friendly scheme based on sponge construction. Specifically, we use $t=12$ Poseidon-128 with $R_F=8$ and $R_P=22$. All experiments are conducted on a laptop computer equipped with an Intel i5-9500F processor, $16$GB RAM, $1$TB SSD storage, and $1$MB/s bandwidth. 

\subsection{Batch Generation}
\noindent \textbf{Hidden state generation time}.
We first analyze the cost of the hidden state. Specifically, we compare the time required for generating the hidden state in two settings: 1) $\mathrm{hidden\_state}_i = H(\mathrm{TX}_{i-2})$ (using one batch), and 2) $\mathrm{hidden\_state}_i = H(\mathrm{TX}_{i-2}, \mathrm{TX}_{k_1}, \cdots, \mathrm{TX}_{k_n})$ (using multiple batches). The results are depicted in \figref{hs time}.

In the one-batch case (top figure), an L2 node can efficiently compute the $\mathrm{hidden\_state}_i$ within $13$ms with a $\mathrm{TX}_{i-2}$ of $1$MB size. Moreover, the generation time does not increase significantly even with larger $\mathrm{TX}_{i-2}$ sizes. When considering multiple batches (bottom figure), the generation time does increase as more previous batches are involved. For simplicity, we fix the size of $\mathrm{TX}_{i-2}$ and $\mathrm{TX}_{k_j}$ to $1$MB. Our design is capable of efficiently handling $10$ batches within $50$ms. To minimize overhead, we have adopted the one-batch design in our system, ensuring that the hidden state only incurs a minor impact.

\noindent \textbf{Available time for building the next batch}.
In our system, we enable the parallel processing of downloading $\mathrm{TX}_{i-1}$ and building $\text{\em batch}_i$ by setting $\mathrm{hidden\_state}_i = H(\mathrm{TX}_{i-2})$. In contrast, existing L2 solutions like zkSync require L2 nodes to download $\mathrm{TX}_{i-1}$ before constructing $\text{\em batch}_i$. We compare the available building time between the two approaches in \figref{ava time}. As expected, the parallel processing in our design significantly increases the available building time. When the batch generation period is $12$s, the available time remains constant at around $12$s regardless of the size of $\mathrm{TX}_{i-1}$. This is because synchronizing the states from L1 is a highly efficient process, and it does not significantly impact the available time. On the other hand, in traditional sequential processing schemes like zkSync, the available time decreases noticeably as the size of $\mathrm{TX}_{i-1}$ increases. This demonstrates that our hidden state design ensures downloading and alleviates the bandwidth burden at the same time.

\begin{table}[]
\centering
\caption{Probability of detecting malicious builder by data availability challenge based on the different number of challenges and proportion of deleted historical data.}
\renewcommand{\arraystretch}{1.5}
\label{DA challenge}
\begin{tabular}{ccccccc}
\cline{2-7} 
\multirow{2}{*}{} &
\multicolumn{6}{c}{\textbf{Proportion of deleted historical data}} \\ \cline{2-7} 
 &
  \multicolumn{1}{c|}{5\%} &
  \multicolumn{1}{c|}{10\%} &
  \multicolumn{1}{c|}{15\%} &
  \multicolumn{1}{c|}{20\%} &
  \multicolumn{1}{c|}{25\%} &
  30\% \\ \hline\hline
\multicolumn{1}{c||}{$s=6$}  & \multicolumn{1}{c|}{23.2\%} & \multicolumn{1}{c|}{41.3\%} & \multicolumn{1}{c|}{57.1\%} & \multicolumn{1}{c|}{68.6\%} & \multicolumn{1}{c|}{77.6\%} & 83.7\% \\ \hline
\multicolumn{1}{c||}{$s=10$} & \multicolumn{1}{c|}{37.7\%} & \multicolumn{1}{c|}{67.6\%} & \multicolumn{1}{c|}{81.6\%} & \multicolumn{1}{c|}{88.2\%} & \multicolumn{1}{c|}{94.8\%} & 97.5\% \\ \hline
\multicolumn{1}{c||}{$s=30$} & \multicolumn{1}{c|}{77.2\%} & \multicolumn{1}{c|}{96.1\%} & \multicolumn{1}{c|}{98.9\%} & \multicolumn{1}{c|}{99.9\%} & \multicolumn{1}{c|}{100\%}  & 100\%  \\ \hline
\multicolumn{1}{c||}{$s=50$} &
  \multicolumn{1}{c|}{94.4\%} &
  \multicolumn{1}{c|}{99.5\%} &
  \multicolumn{1}{c|}{100\%} &
  \multicolumn{1}{c|}{100\%} &
  \multicolumn{1}{c|}{100\%} &
  100\% \\ \hline
\end{tabular}%
\end{table}

\begin{figure}
	\centering
    \includegraphics[width=1\columnwidth]{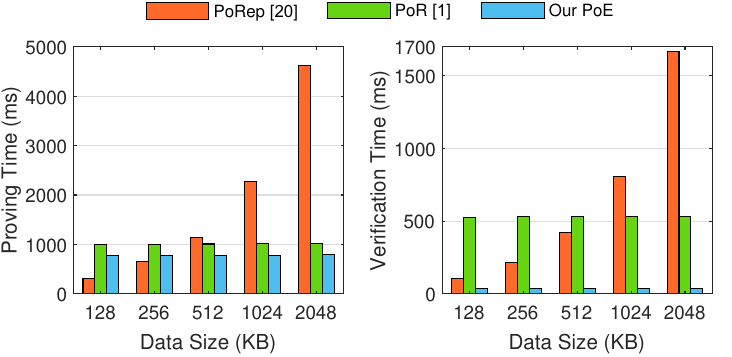}
    \caption{Efficiency comparison of our proof of existence (PoE) with other storage proofs.}
	\label{poe_time}
\end{figure}

\begin{figure*}
	\centering
	\begin{minipage}{0.245\linewidth}
		\centering
		\includegraphics[width=1\textwidth]{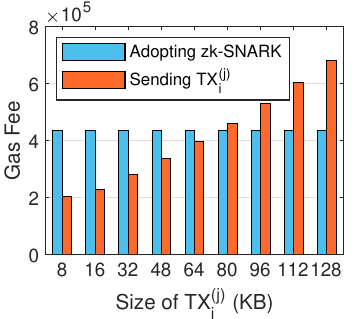}
		\caption{Gas fee.}
		\label{gas fee}     
	\end{minipage}
	\begin{minipage}{0.245\linewidth}
		\centering
		\includegraphics[width=1\textwidth]{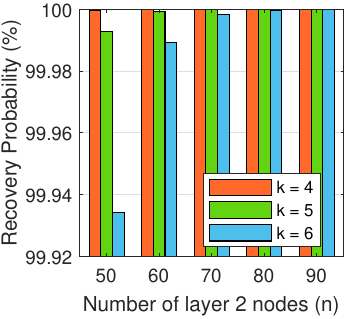}
		\caption{No failure.}
		\label{recovery rate no failure}     
	\end{minipage}
    \begin{minipage}{0.245\linewidth}
		\centering
		\includegraphics[width=1\textwidth]{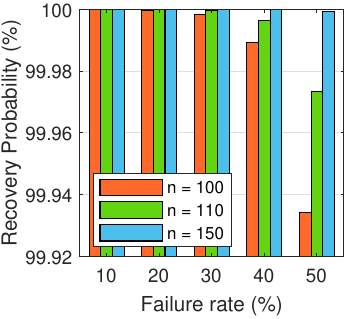}
		\caption{$k = 5$.}
		\label{recovery rate on n}     
	\end{minipage}
	\begin{minipage}{0.245\linewidth}
		\centering
		\includegraphics[width=1\textwidth]{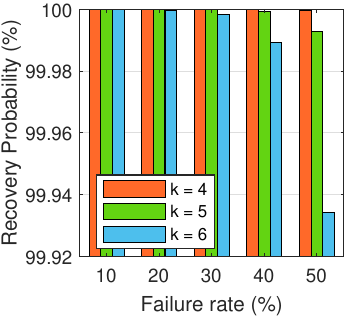}
		\caption{$n = 100$.}
		\label{recovery rate on k}     
	\end{minipage}
\end{figure*}



\subsection{Data Availability Challenge}
We extensively evaluate the performance of the data availability challenges. The experiment shows that the data availability challenge can effectively identify malicious builders who delete historical transactions.

\noindent \textbf{Probability of identifying malicious builders}.
We evaluate the performance of our data availability challenge by analyzing the probability of detecting malicious builders who delete parts of historical data. Assuming each builder stores 10,000 batches, \tabref{DA challenge} illustrates the probability of detection based on the number of challenges ($s$) and the proportion of deletion. The results indicate that our data availability challenge can effectively identify malicious builders even when they delete only a small portion of historical data. For instance, if a builder deletes $5\%$ of the historical data, there is a $94.4\%$ chance of detection with $50$ times of challenges. It is important to note that $s = 50$ is a conventional value in real-world scenarios where builders are less decentralized. Additionally, deleting $5\%$ of the historical data does not alleviate the storage burden on builders. As more data is deleted, the probability of detection increases significantly. For example, when $30\%$ data is deleted, there is a $97.5\%$ chance of being caught with only $10$ challenges.



\noindent \textbf{Efficiency of data availability challenge}.
Since builders need to generate a zk-SNARK proof for data availability challenges, we further compare the proving and verification time of our proof of existence with other storage proofs such as PoR \cite{anthoine2021dynamic} and PoRep \cite{fisch2018poreps} in \figref{poe_time}, considering different sizes of challenged historical data (i.e., different sizes of $\mathrm{TX}_i^{(j)}$). Remarkably, regardless of the data size, the proving time of our scheme remains approximately $780$ms, while the verification time is around $35$ms. Compared with other storage proofs, our scheme is the most efficient one, especially for large data. These results indicate our proof of efficiency is more suitable for real-world applications of L2. Moreover, the results show that the time required in our scheme does not increase significantly with the data size (unlike the PoRep, which increases linearly with the data size). The above characteristics make our design highly suitable for Ethereum scaling solutions.


\noindent \textbf{Gas fees}.
In the design of the data availability challenges, we highlight that directly responding with $\mathrm{TX}_i^{(j)}$ is suitable for low transaction rate cases, while adopting a zk-SNARK is more suitable for scaling solutions. To compare the gas fees of these two designs, we consider the gas fee cost of each approach. To deploy the arbiter contract, the straightforward approach incurs a gas fee of 1,769,307, while the zk-SNARK solution costs 3,545,788 gas. Although deploying zk-SNARK is more expensive, it is considered affordable in real-world systems as it is a one-time cost.

For each data availability challenge, the gas fee is shown in \figref{gas fee}. As expected, the fee for the zk-SNARK solution remains consistent at around 434,900 gas, regardless of the size of $\mathrm{TX}_i^{(j)}$. On the other hand, the gas fee for responding with $\mathrm{TX}_i^{(j)}$ directly increases linearly with the size of $\mathrm{TX}_i^{(j)}$. When the size of $\mathrm{TX}_i^{(j)}$ exceeds $80$KB, the zk-SNARK design becomes more cost-effective and is better suited for scaling solutions.

\subsection{Partial storage}
In our partial storage design, it is important to ensure a high probability of data recovery in the L2 network, even with some node failures. Let $n$ denote the number of builders in the L2 network. In our partial storage design, we divide each batch into $k$ parts. We present the probability of data recovery under different values of $n$ and $k$ without considering failures, as shown in \figref{recovery rate no failure}. As expected, the recovery probability increases with $n$ and decreases with $k$ since a larger $n$ and a smaller $k$ indicate a more robust network. The results also show our design is acceptable for real-world scenarios, with the recovery probability remaining above $99.99\%$ when $k = 5$ for a network of $50$ builders.

Moreover, we examine the recovery probability under different failure rates. The results, with a fixed $k = 5$, are depicted in \figref{recovery rate on n}, while those with a fixed $n = 100$ are shown in \figref{recovery rate on k}. We observe that a high recovery probability ($> 99.9\%$) can be maintained even with $50\%$ failure rate. When the failure rate is below $30\%$, the recovery probability almost approaches $100\%$ for all cases. Notably, our experiment also considers partial storage, meaning that each builder only needs to store $1/k$ of the historical data. These results emphasize that our partial storage design ensures both data availability and alleviates the storage burden simultaneously.




\begin{table}[]
\caption{The difficulty ratio of a malicious builder when colluding with different proportions of proposers. INF means the malicious builder can never find a proper $r$.}

\label{PoL}
\renewcommand{\arraystretch}{1.5}
\centering
\resizebox{1\columnwidth}{!}{%
\begin{tabular}{cccccccc}
\cline{2-8}
\multicolumn{1}{l}{}        & \multicolumn{7}{c}{\textbf{Proportion of Colluded Proposers}}                                                                                                                                                                                                                 \\ \cline{2-8} 
\multicolumn{1}{c}{}      & \multicolumn{1}{c|}{$1\%$} & \multicolumn{1}{c|}{$2\%$}                & \multicolumn{1}{c|}{$3\%$}                & \multicolumn{1}{c|}{$5\%$}                & \multicolumn{1}{c|}{$10\%$}               & \multicolumn{1}{c|}{$20\%$}              & $30\%$      \\ \hline\hline
\multicolumn{1}{c||}{$a=1.5$}  & \multicolumn{1}{c|}{INF}   & \multicolumn{1}{c|}{INF}                  & \multicolumn{1}{c|}{$4.1 \times 10^{65}$} & \multicolumn{1}{c|}{$6.9 \times 10^{36}$} & \multicolumn{1}{c|}{$1.5 \times 10^{15}$} & \multicolumn{1}{c|}{$1.8 \times 10^{4}$} & $5.1$       \\ \hline
\multicolumn{1}{c||}{$a=2.5$} & \multicolumn{1}{c|}{INF}   & \multicolumn{1}{c|}{INF}                  & \multicolumn{1}{c|}{$2.0 \times 10^{61}$} & \multicolumn{1}{c|}{$2.7 \times 10^{32}$} & \multicolumn{1}{c|}{$2.7 \times 10^{6}$}  & \multicolumn{1}{c|}{$1.8$}               & $1.0002$    \\ \hline
\multicolumn{1}{c||}{$a=5.5$} & \multicolumn{1}{c|}{INF}   & \multicolumn{1}{c|}{INF}                  & \multicolumn{1}{c|}{$2.9 \times 10^{48}$} & \multicolumn{1}{c|}{$3.4 \times 10^{19}$} & \multicolumn{1}{c|}{$1.0062$}             & \multicolumn{1}{c|}{$\approx 1$}         & $\approx 1$ \\ \hline
\multicolumn{1}{c||}{$a=10.5$}& \multicolumn{1}{c|}{INF}   & \multicolumn{1}{c|}{$2.9 \times 10^{62}$} & \multicolumn{1}{c|}{$3.8 \times 10^{26}$} & \multicolumn{1}{c|}{$1.0069$}             & \multicolumn{1}{c|}{$\approx 1$}          & \multicolumn{1}{c|}{$\approx 1$}         & $\approx 1$ \\ \hline
\end{tabular}
}
\end{table}

\subsection{Against the MEV Attack}
Finally, we evaluate the performance of our system against MEV attacks. We define the monotone decreasing function $D$ described in \secref{sec: Proof of Luck and Period Separation} as $\frac{b \cdot 2^{256}}{1 + e^{10 \abs{\mathrm{Proposer}_j - \mathrm{luck}_{i}} - a}}$, where $a$ and $b$ are predefined values that control the difficulty. Specifically, $a$ determines the expected number of valid proposers. For example, by setting $a = 10.5$, we can ensure that the $10$ closest proposers to $\mathrm{luck}_{i}$ are valid (the expected number is $a - 0.5$). $b$ is related to the probability of a build finding a proper $r$ in one attempt. For instance, if $b = 0.1$, an honest build will have a $10\%$ chance of finding the proper $r$ each time. To compare the difficulty of a malicious build with that of an honest one, we use the difficulty rate between the malicious builder and the honest builder, given by $\frac{D(|\mathrm{Proposer}_{j'} - \mathrm{luck}_{i}|)}{D(|\mathrm{Proposer}_{j} - \mathrm{luck}_{i}|)}$, where $\mathrm{Proposer}_{j}$ and $\mathrm{Proposer}_{j'}$ are proposals chosen by the honest and malicious builder, respectively. In our experiment, we consider $1000$ proposals and the malicious builder colludes with different proportions of proposers. Since $b$ is eliminated in the difficulty rate, we focus on different settings of $a$ as shown in \tabref{PoL}. The term ``INF'' indicates that the valid range is less than $1$, and the probability of finding an $r$ such that $\mathrm{Hash}(\text{\em batch}_i.\mathrm{head}(r)) = 0^{256}$ is negligible.

The difficulty rate increases significantly as the proportion and $a$ decrease. For example, if the builder colludes with $2\% \times 1000 = 20$ proposers and $a = 10.5$, the probability of finding $r$ decreases by $2.9 \times 10^{62}$. It is worth noting that as the builder colludes with more proposers, the difficulty rate increases to nearly $1$. This is expected since when $a$ is large, the probability of a colluded proposal being close to the lucky number is overwhelming, indicating that the proposal is valid. Thus, this observation holds true regardless of the function $D$ we choose. However, in real-world scenarios, we consider such collusion unlikely as the proposer network should be decentralized. Colluding with more than $30\%$ of the proposers would be impractical. Therefore, our designs for proof of luck and period separation effectively reduce collusion between builders and proposers and provide efficient protection against MEV attacks.




\section{Discussion}
\label{sec: Discussion}
We discuss some technical details and concerns in our design.

\noindent \textbf{Sharing the hidden state}. In our design, it is possible for L2 nodes to collaborate and share the latest hidden state with each other, resembling a mining pool-like collaboration. This collaboration is allowed because there are certain L2 nodes that download complete transactions to generate the hidden state. As a result, all historical transactions remain complete within the L2 network.

\noindent \textbf{DDoS attacks in proof of existence}. In the proof of existence, L2 nodes are required to respond promptly to data availability challenges. This could potentially expose L2 nodes to DDoS attacks, where an attacker could repeatedly challenge an L2 node, forcing it to spend high gas fees for the responses. Our design mitigates high gas fees for a single challenge. For multiple challenges, we can employ the batch verification of KZG and require the challenger to provide some ``query fee'' that will be transferred to the L2 node upon successful response to the challenge. 

\noindent \textbf{Separation between proof of download and proof of existence}. The separate designs of proof of download and proof of existence offer several benefits. First, proof of download can solve the lazy validator problem and partially address the historical data storage problem. Second, the proof of download imposes additional costs on malicious L2 nodes. If they choose to delete historical data to reduce storage costs, they will have to incur higher bandwidth costs to re-download in response to data availability challenges. Third, a \textit{deterministic} proof of download can be theoretically inferred from a \textit{deterministic} proof of existence (which involves checking all data). However, a deterministic proof of existence requires at least $O(n)$ time. For the sake of system efficiency, we separate these two designs by employing a probabilistic proof of existence and a deterministic proof of download. Consequently, we ensure better data availability than merely using probabilistic proof of existence.

\noindent \textbf{DAS vs. proof of download.}
DAS in Danksharding \cite{Danksharding} allows L1 nodes to randomly sample data to test the availability and punishes the L2 batch builder that fails in the sampling test. With the help of erasure coding and KZG commitment, users can rebuild the original data even with some failures. We summarize the difference of our proof of download as follows. First, L1 nodes adopt DAS to check blob data availability, which incurs the lazy validator problem. In our scheme, the L2 batch builder must propagate the newly mined batch, and other L2 nodes must verify data availability (download the previous batch and generate a valid proof of download) before generating the latest batch. Consequently, L2 nodes cannot benefit (generate a new batch) if they are too lazy to validate. Second, as the blobs will be deleted after a period, it is critical to ensure L2 nodes download transactions promptly. DAS cannot address this issue, while our proof of download solves this problem by restricting their ability to generate new batches without previous transactions.

\noindent \textbf{PBS vs. role separation.}
PBS in Danksharding separates L1 nodes into different roles with a proposer/builder separation: less decentralized builders with high bandwidth to build and broadcast the body of batches and more decentralized proposers with low bandwidth to build the header of batches based on the selected body. As PBS in Danksharding works on L1, the cost of L2 is not affected. But our role separation works on L2, which can significantly reduce the hardware requirements of L2 proposers. Besides, PBS in Danksharding cannot fully avoid MEV attacks when builders collude with proposers, while our scheme can prevent MEV attacks even considering collusion.

\section{Conclusion}
Data availability and decentralization are two critical requirements in the design of L2 solutions. This paper introduces novel techniques, namely proof of download, proof of existence, and role separation, tailored for L2 networks. Proof of download and proof of existence ensure data availability by requiring L2 nodes to download the latest batch and punishing nodes that delete historical transactions. Role separation reduces the hardware requirement of joining L2 networks, resulting in a more decentralized system capable of mitigating MEV attacks. Experimental results show our system can ensure security and efficiency at the same time.

\section{Acknowledgments}
This research has received partial support from HK RGC GRF under Grants PolyU 15216721, 15207522, 15202123, NSFC Youth 62302418, and a donation from Huawei P0048565.

\bibliographystyle{IEEEtranS}
\bibliography{Layer2_NDSS}

\begin{thebibliography}{10}
\providecommand{\url}[1]{#1}
\csname url@samestyle\endcsname
\providecommand{\newblock}{\relax}
\providecommand{\bibinfo}[2]{#2}
\providecommand{\BIBentrySTDinterwordspacing}{\spaceskip=0pt\relax}
\providecommand{\BIBentryALTinterwordstretchfactor}{4}
\providecommand{\BIBentryALTinterwordspacing}{\spaceskip=\fontdimen2\font plus
\BIBentryALTinterwordstretchfactor\fontdimen3\font minus \fontdimen4\font\relax}
\providecommand{\BIBforeignlanguage}[2]{{%
\expandafter\ifx\csname l@#1\endcsname\relax
\typeout{** WARNING: IEEEtran.bst: No hyphenation pattern has been}%
\typeout{** loaded for the language `#1'. Using the pattern for}%
\typeout{** the default language instead.}%
\else
\language=\csname l@#1\endcsname
\fi
#2}}
\providecommand{\BIBdecl}{\relax}
\BIBdecl

\bibitem{karame2016security}
G.~Karame, ``On the security and scalability of bitcoin's blockchain,'' in \emph{Proc. of the ACM Conference on Computer \& Communications Security (CCS)}, 2016, pp. 1861--1862.

\bibitem{nakamoto2008bitcoin}
S.~Nakamoto, ``Bitcoin: A peer-to-peer electronic cash system,'' \emph{Decentralized Business Review}, p. 21260, 2008.

\bibitem{wood2014ethereum}
G.~Wood \emph{et~al.}, ``Ethereum: A secure decentralised generalised transaction ledger,'' \emph{Ethereum project yellow paper}, vol. 151, no. 2014, pp. 1--32, 2014.

\bibitem{gramoli2015rollup}
V.~Gramoli, L.~Bass, A.~Fekete, and D.~W. Sun, ``Rollup: Non-disruptive rolling upgrade with fast consensus-based dynamic reconfigurations,'' \emph{IEEE Transactions on Parallel and Distributed Systems (TPDS)}, vol.~27, no.~9, pp. 2711--2724, 2015.

\bibitem{buterin2018chain}
V.~Buterin, ``\protect{On-chain Scaling to Potentially 500 Tx/sec Through Mass Tx Validation},'' \url{https://ethresear.ch/t/on-chain-scaling-to-potentially-500-tx-sec-through-mass-tx-validation/3477}, 2018.

\bibitem{Floersch2019Ethereum}
K.~Floersch, ``\protect{Ethereum Smart Contracts in L2: Optimistic Rollup},'' \url{https://medium.com/plasma-group/ethereum-smart-contracts-in-l2-optimistic-rollup-2c1cef2ec537}, 2019.

\bibitem{Optimism}
O.~Foundation, ``\protect{Optimism},'' \url{https://www.optimism.io/}, 2022.

\bibitem{ArbitrumOne}
O.~Labs, ``\protect{Arbitrum One},'' \url{https://portal.arbitrum.one/}, 2022.

\bibitem{Polygonscan}
Polygonscan, ``\protect{Polygon},'' \url{https://polygonscan.com/}, 2021.

\bibitem{zkSync}
M.~Labs, ``\protect{zkSync},'' \url{https://zksync.io/}, 2021.

\bibitem{Validium}
Ethereum.org, ``\protect{Validium},'' \url{https://ethereum.org/en/developers/docs/scaling/validium/}, 2021.

\bibitem{Volition}
Polynya, ``\protect{Volitions: Best of all Worlds},'' \url{https://polynya.medium.com/volitions-best-of-all-worlds-cfd313aec9a8}, 2021.

\bibitem{Loopring}
L.~P. Ltd, ``\protect{Loopring},'' \url{https://loopring.org/\#/}, 2021.

\bibitem{zkPorter}
M.~Labs, ``\protect{zkPorter: A Breakthrough in L2 Scaling},'' \url{https://blog.matter-labs.io/zkporter-a-breakthrough-in-l2-scaling-ed5e48842fbf}, 2021.

\bibitem{TrustModels}
V.~Buterin, ``\protect{Trust Models},'' \url{https://vitalik.eth.li/general/2020/08/20/trust.html}, 2020.

\bibitem{ateniese2014proofs}
G.~Ateniese, I.~Bonacina, A.~Faonio, and N.~Galesi, ``Proofs of space: When space is of the essence,'' in \emph{Proc. of the International Conference on Security and Cryptography for Networks (SCN)}.\hskip 1em plus 0.5em minus 0.4em\relax Springer, 2014, pp. 538--557.

\bibitem{boneh2020efficient}
D.~Boneh, J.~Drake, B.~Fisch, and A.~Gabizon, ``Efficient polynomial commitment schemes for multiple points and polynomials,'' \emph{IACR Cryptology ePrint Archive}, 2020.

\bibitem{HybridLayer2}
V.~Buterin, ``\protect{The Dawn of Hybrid Layer 2 Protocols},'' \url{https://alidevjimmy.github.io/general/2019/08/28/hybrid\_layer\_2.html}, 2019.

\bibitem{Average_Block_Size}
Etherscan, ``\protect{Ethereum Average Block Size Chart},'' \url{https://etherscan.io/chart/blocksize}, 2019.

\bibitem{Filecoin}
Filecoin, ``\protect{Filecoin Proving Subsystem},'' \url{https://github.com/filecoin-project/rust-fil-proofs}, 2019.

\bibitem{Data_Withholding_Attacks}
Ethereum, ``\protect{What Is Validium},'' \url{https://ethereum.org/developers/docs/scaling/validium#what-is-validium}, 2019.

\bibitem{fisch2018poreps}
B.~Fisch, ``Poreps: Proofs of space on useful data,'' \emph{IACR Cryptology ePrint Archive}, 2018.

\bibitem{dziembowski2015proofs}
S.~Dziembowski, S.~Faust, V.~Kolmogorov, and K.~Pietrzak, ``Proofs of space,'' in \emph{Proc. of the Annual International Cryptology Conference (CRYPTO)}.\hskip 1em plus 0.5em minus 0.4em\relax Springer, 2015, pp. 585--605.

\bibitem{chu2018curses}
S.~Chu and S.~Wang, ``The curses of blockchain decentralization,'' \emph{IACR Cryptology ePrint Archive}, 2018.

\bibitem{ateniese2020proof}
G.~Ateniese, L.~Chen, M.~Etemad, and Q.~Tang, ``Proof of storage-time: Efficiently checking continuous data availability,'' \emph{IACR Cryptology ePrint Archive}, 2020.

\bibitem{anthoine2021dynamic}
G.~Anthoine, J.-G. Dumas, M.~de~Jonghe, A.~Maignan, C.~Pernet, M.~Hanling, and D.~S. Roche, ``Dynamic proofs of retrievability with low server storage,'' in \emph{Proc. of the USENIX Security Symposium (Security)}, 2021, pp. 537--554.

\bibitem{zhou2021just}
L.~Zhou, K.~Qin, A.~Cully, B.~Livshits, and A.~Gervais, ``On the just-in-time discovery of profit-generating transactions in defi protocols,'' in \emph{Proc. of the IEEE Symposium on Security and Privacy (S\&P)}.\hskip 1em plus 0.5em minus 0.4em\relax IEEE, 2021, pp. 919--936.

\bibitem{qin2022quantifying}
K.~Qin, L.~Zhou, and A.~Gervais, ``Quantifying blockchain extractable value: How dark is the forest?'' in \emph{Proc. of the IEEE Symposium on Security and Privacy (S\&P)}.\hskip 1em plus 0.5em minus 0.4em\relax IEEE, 2022, pp. 198--214.

\bibitem{EIP4844}
E.~I. Proposals, ``\protect{EIP-4844: Shard Blob Transactions},'' \url{https://eips.ethereum.org/EIPS/eip-4844}, 2022.

\bibitem{poon2017plasma}
J.~Poon and V.~Buterin, ``\protect{Plasma: Scalable Autonomous Smart Contracts},'' \url{https://www.plasma.io/plasma-deprecated.pdf}, 2022.

\bibitem{Danksharding}
D.~Feist, ``\protect{New Sharding Design with Tight Beacon and Shard Block Integration},'' \url{https://notes.ethereum.org/@dankrad/new\_sharding}, 2022.

\bibitem{EthereumDanksharding}
E.~Community, ``\protect{Danksharding},'' \url{https://ethereum.org/en/roadmap/danksharding/#danksharding}, 2023.

\bibitem{kate2010constant}
A.~Kate, G.~M. Zaverucha, and I.~Goldberg, ``Constant-size commitments to polynomials and their applications,'' in \emph{Proc. of the Annual International Conference on the Theory and Application of Cryptology and Information Security (ASIACRYPT)}.\hskip 1em plus 0.5em minus 0.4em\relax Springer, 2010, pp. 177--194.

\bibitem{zhang2022polynomial}
J.~Zhang, T.~Xie, T.~Hoang, E.~Shi, and Y.~Zhang, ``Polynomial commitment with a $\{$One-to-Many$\}$ prover and applications,'' in \emph{Proc. of the USENIX Security Symposium (Security)}, 2022, pp. 2965--2982.

\bibitem{sguanci2021layer}
C.~Sguanci, R.~Spatafora, and A.~M. Vergani, ``Layer 2 blockchain scaling: A survey,'' \emph{IACR Cryptology ePrint Archive}, 2021.

\bibitem{chen2023hyperplonk}
B.~Chen, B.~B{\"u}nz, D.~Boneh, and Z.~Zhang, ``Hyperplonk: Plonk with linear-time prover and high-degree custom gates,'' in \emph{Proc. of the Annual International Conference on the Theory and Applications of Cryptographic Techniques (EUROCRYPT)}.\hskip 1em plus 0.5em minus 0.4em\relax Springer, 2023, pp. 499--530.

\bibitem{dziembowski2018fairswap}
S.~Dziembowski, L.~Eckey, and S.~Faust, ``Fairswap: How to fairly exchange digital goods,'' in \emph{Proc. of the ACM Conference on Computer \& Communications Security (CCS)}, 2018, pp. 967--984.

\bibitem{boneh2021halo}
D.~Boneh, J.~Drake, B.~Fisch, and A.~Gabizon, ``Halo infinite: Proof-carrying data from additive polynomial commitments,'' in \emph{Proc. of the Annual International Cryptology Conference (CRYPTO)}, 2021, pp. 649--680.

\bibitem{gabizon2019plonk}
A.~Gabizon, Z.~J. Williamson, and O.~Ciobotaru, ``Plonk: Permutations over lagrange-bases for oecumenical noninteractive arguments of knowledge,'' \emph{IACR Cryptology ePrint Archive}, 2019.

\bibitem{ZKFair}
Z.~Ltd., ``\protect{ZKFair},'' \url{https://docs.zkfair.io/}, 2024.

\bibitem{ImmutableX}
S.~Ltd., ``\protect{Immutable X},'' \url{https://docs.starkware.co/starkex/perpetual/perpetual\_overview.html}, 2024.

\bibitem{dYdX}
dYdX Ltd., ``\protect{dYdX},'' \url{https://dydx.l2beat.com/}, 2024.

\bibitem{Base}
B.~Ltd., ``\protect{Base},'' \url{https://docs.base.org/}, 2024.

\bibitem{Mode}
M.~Ltd., ``\protect{Mode},'' \url{https://docs.mode.network/introduction/readme}, 2024.

\bibitem{PBS}
E.~Community, ``\protect{Proposer-builder separation},'' \url{https://ethereum.org/en/roadmap/pbs/}, 2024.

\end{thebibliography}



\end{document}